\documentclass[11pt]{amsart}
\usepackage[utf8]{inputenc}
\usepackage{amsmath, amsthm, amssymb}
\usepackage[nobysame,abbrev,alphabetic]{amsrefs}
\usepackage{graphicx}
\usepackage[usenames]{color}
\usepackage{srcltx} 
\usepackage{tikz-cd}
\usepackage{tikz}
\usetikzlibrary{positioning}
\usepackage{mathtools}

\newtheorem{thm}{Theorem}[section]
\newcommand{\bt}{\begin{thm}}
\newcommand{\et}{\end{thm}}

\newtheorem{ex}[thm]{Example}

\newtheorem{cor}[thm]{Corollary}   
\newcommand{\bc}{\begin{cor}}
\newcommand{\ec}{\end{cor}}

\newtheorem{lem}[thm]{Lemma}   
\newcommand{\bl}{\begin{lem}}
\newcommand{\el}{\end{lem}}

\newtheorem{prop}[thm]{Proposition}
\newcommand{\bp}{\begin{prop}}
\newcommand{\ep}{\end{prop}}

\newtheorem{defn}[thm]{Definition}
\newcommand{\bd}{\begin{defn}}    
\newcommand{\ed}{\end{defn}}

\newtheorem{rmrk}[thm]{Remark}   

\newcommand{\br}{\begin{rmrk}}
\newcommand{\er}{\end{rmrk}}

\newcommand{\N}{\mathbb{N}}

\newcommand{\R}{\mathbb{R}}

\DeclareMathOperator{\dis}{dis}

\newcommand{\diam}{\operatorname{Diam}}

\def\XXint#1#2#3{{\setbox0=\hbox{$#1{#2#3}{\int}$ }
\vcenter{\hbox{$#2#3$ }}\kern-.6\wd0}}

\begin{document}

\title[Null Distance and Convergence of Warped Product Spacetimes]{Null Distance and Gromov-Hausdorff Convergence of Warped Product Spacetimes}

\author{Brian Allen}
\address[CUNY]{Lehman College, CUNY}
\urladdr{\url{https://sites.google.com/view/brian-allen}}

\begin{abstract}
 What is the analogous notion of Gromov-Hausdorff convergence for sequences of spacetimes? Since a Lorentzian manifold is not inherently a metric space, one cannot simply use the traditional definition. One approach offered by Sormani and Vega \cite{SV} is to define a metric space structure on a spacetime by means of the null distance. Then one can define convergence of spacetimes using the usual definition of Gromov-Hausdorff convergence. In this paper we explore this approach by giving many examples of sequences of warped product spacetimes with the null distance converging in the Gromov-Hausdorff sense.  In addition, we give an optimal convergence theorem which shows that under natural geometric hypotheses a sequence of warped product spacetimes converge to a specific limiting warped product spacetime. The examples given further serve to show that the hypotheses of this convergence theorem are optimal.
\end{abstract}

\maketitle

\section{Introduction}

 What is the analogous notion of Gromov-Hausdorff convergence for sequences of spacetimes? Since a Lorentzian manifold is not inherently a metric space, one cannot simply use the traditional definition. This problem was initially taken up by J. Noldus \cite{N}, which has been more recently addressed by E. Minguzzi and S. Suhr \cite{MS} and by O. M\"{u}ller \cite{M}. Another approach offered by Sormani and Vega \cite{SV} is to define a metric space structure on a spacetime by means of the null distance. Then one can define convergence of spacetimes using the usual definition of Gromov-Hausdorff convergence for metric spaces. In this paper we explore this approach involving the null distance and give many examples which can be used to further compare and contrast all approaches.

When one considers a sequence of spacetimes equipped with the null distance it is natural to ask for geometric hypotheses which can guarantee convergence to a particular null distance spacetime. In the case of sequences of Riemannian manifolds, the author, R. Perales, and C. Sormani \cite{Allen-Perales-Sormani} provided such a convergence theorem for Sormani-Wenger Intrinsic Flat convergence. The theorem says that if the diameter of the sequence is bounded, the volume converges to the desired limit volume, and distances converge from below then the sequence converges to the desired limiting Riemannian manifold in the volume preserving Sormani-Wenger Intrinsic Flat sense. Many illuminating examples were given by the author and C. Sormani \cite{Allen-Sormani, Allen-Sormani-2} which justified the hypotheses. Our goal in this paper is to continue the analogous investigation for null distance spacetimes by providing many probing examples and establishing convergence theorems for warped product spacetimes under optimal hypotheses.

The first such theorem shows that warped product spacetimes equipped with the null distance, whose warping function $f_j$ is larger than a warping function $f_{\infty}$, and so that $f_j$ converges to $f_{\infty}$ in $L^1$, converge in the uniform sense, i.e $C^0$ convergence of distance functions on $M \times M$ (defined in section \ref{sect-Background}). Note that for a spacetime $(M,g_j)$ with a time function $\tau:M \rightarrow \R$ we denote the corresponding null distance by $\hat{d}_{\tau,g_j}$, defined in section \ref{sect-Background}.

\begin{thm}\label{thm-MainTheorem tau}
Let $(\Sigma^n,\sigma)$ be a compact, connected Riemannian manifold, $M=[t_0,t_1]\times \Sigma$, $f_j:[t_0,t_1]\rightarrow (0,\infty)$ continuous, $g_j=-dt^2+f_j(t)^2\sigma$, $j \in \N\cup\{\infty\}$. If $\tau(t)=\int_0^t \frac{1}{f_{\infty}(r)}dr$,
\begin{align}
\int_{t_0}^{t_1}|f_j-&f_{\infty}|dt \rightarrow 0,
  \\ f_j(t) &\ge  f_{\infty}(t), \quad \forall t \in [t_0,t_1],\label{eq lower bound}
\end{align}
then
\begin{align}
    (M,\hat{d}_{\tau,g_j})\rightarrow (M,\hat{d}_{\tau,g_{\infty}}),
\end{align}
uniformly. 
\end{thm}

\begin{rmrk}
    One should notice that $\tau(t)=\int_0^t \frac{1}{f_{\infty}(r)}dr$ is differentiable $\tau'(t)=\frac{1}{f_{\infty}(t)}$, $|\tau'|\le C$, and $\nabla \tau = \frac{1}{f_{\infty}(t)} \nabla t$ with respect to $g_0=-dt^2+\sigma$. Hence by Proposition 4.12 of C. Sormani and C. Vega \cite{SV} we see that $\tau$ is anti-Lipshitz (see Definition \ref{def-AntiLipschitz}) and defines a definite null distance. By Theorem 1.9 of A. Burtscher and  L. Garc\`{ı}a-Heveling \cite{BG2} we also know that $\tau$ encodes causality, i.e. $p \le q$ if and only if $\hat{d}_{\tau,g}(p,q)=|\tau(p)-\tau(q)|$.
\end{rmrk}

 Our goal in the next theorem is to have the convergence of the sequence of null distance warped product spacetimes with respect to the $t$-coordinate, which is the cosmological time function in the case of warped products (See Lemma \ref{lem-tIsCTF}). Additionally, one would like to relax the lower bound condition \eqref{eq lower bound} to just convergence from below. We also note that the proof of Theorem \ref{thm-MainTheorem t} uses Theorem \ref{thm-MainTheorem tau} and properties of the null distance.

\begin{thm}\label{thm-MainTheorem t}
Let $(\Sigma^n,\sigma)$ be a compact, connected Riemannian manifold, $M=[t_0,t_1]\times \Sigma$, $f_j:[t_0,t_1]\rightarrow (0,\infty)$ continuous, $g_j=-dt^2+f_j(t)^2\sigma$, $j \in \N\cup\{\infty\}$. If
\begin{align}
 \int_{t_0}^{t_1}&|f_j-f_{\infty}|dt \rightarrow 0,
  \\ f_j(t) &\ge \left( 1 -\frac{1}{j} \right)f_{\infty}(t), \quad \forall t \in [t_0,t_1],
\end{align}
then
\begin{align}
    (M,\hat{d}_{t,g_j})\rightarrow (M,\hat{d}_{t,g_{\infty}}),
\end{align}
uniformly. 
\end{thm}

This result can also be interpreted as saying that warped product spacetimes with the null distance, which are close from below in the $C^0$ sense, and close from above in the $L^1$ sense, are sequentially stable with respect to the uniform distance. One may think that requiring $f_j$ to converge to $f_{\infty}$ from below is too strong but it was already shown by the author and A. Burtscher \cite{AB} that this was necessary. We review this example in section \ref{sect-Examples} which shows that if there is even just one point which does not converge from below then the sequence will converge to a metric space which is not the null distance of the limiting warped product spacetime. This is one of a family of examples which shows that one will generally need convergence from below in order to guarantee Gromov-Hausdorff convergence. 

The rest of the examples in section \ref{sect-Examples} are for sequences of metrics on $[0,2]\times\mathbb{D}^n$ with the metric $g_j=-dt^2+f_j(r,t)^2 \sigma$ where $\sigma$ is the flat metric on the disk $\mathbb{D}^n$ and $r$ is the radial coordinate from the origin. These examples help to illuminate what the generalization of Theorem \ref{thm-MainTheorem t} should be but they also show how this generalization is not a straight forward adaption of the VADB result of the author, R. Perales, and C. Sormani \cite{Allen-Perales-Sormani} in the Riemannian case. We now describe each of these examples in order to highlight there importance. The author also believes that these examples should provide important tools for comparing and contrasting other possible definitions of Gromov-Hausdorff convergence of spacetimes such as E. Minguzzi and S. Suhr \cite{MS}, and O. M\"{u}ller \cite{M}, definitions of measured Gromov-Hausdorff convergence by M. Braun \cite{B}, and F. Caveletti and A. Mondino \cite{CM}, and  Spacetime Intrinsic Flat convergence of A. Sakovich and C. Sormani \cite{S-Ob,SSFuture}.

 In Example \ref{Ex-ManyBlowUpsTaxi} we see if $f_{\infty}=1 \le f_j$, but $f_j(t)$ blows up along a dense set of times in $[0,2]$ so that $f_j$ is unbounded in $L^1$, then the blow up will drastically effect the sequence and the limiting metric will be a taxi metric space on $[0,2]\times \mathbb{D}^n$. Similarly, in Example \ref{Ex-ManyShortcutTaxi} we see if $f_{\infty}=1 \le f_j \le 2$ and $f_j$ converges to $1$ pointwise along a dense set of points in $[0,2]$, but $L^1$ converges to $2$, then the sequence of null distance metric spaces will not converge to Minkowski space with the null distance. These examples clearly demonstrate the necessity of $L^1$ convergence in Theorem \ref{thm-MainTheorem t}. 

In the next few examples we want to test the effects of different blow up rates of $f_j(r,t)$ on the sequence of null distance metrics for spacetimes where $f_j$ is allowed to depend on $r$. In Example \ref{Ex-BlowUpBubble}, we see that for a precise blow up rate we find bubbling phenomenon along the sequence. In particular, the limiting metric space has $n+1$ Hausdorff dimension throughout but is not a manifold. In Example \ref{Ex-BlowUpSpline}, we see that for a precise blow up rate we find a generalized spline forms along the sequence. In particular, this is another example where the Gromov-Hausdorff limit and Spacetime Intrinsic Flat limit of a sequence of null distance spacetimes do not agree (See Example 5.11 in \cite{AB}). In Example \ref{Ex-BlowupRate}, we see that for blow up rates below the critical rate we find convergence of the sequences of null distance warped products to Minkowski space with the null distance. These examples give direction as to what Theorem \ref{thm-MainTheorem t} should be for Gromov-Hausdorff and Spacetime Intrinsic Flat convergence studied on more general classes of spacetimes.

In section \ref{sect-Background}, we provide the necessary background on the null distance to be able to understand the examples, theorems,  and proofs of this paper. In section \ref{sect-Examples}, we provide many examples which justify the hypotheses of Theorem \ref{thm-MainTheorem tau} and Theorem \ref{thm-MainTheorem t}, as well as provide hints at what the general theorem for Spacetime Intrinsic Flat convergence should be. In section \ref{sect-Proofs}, we provide the proofs of the main theorems.

\textbf{Acknowledgements:} The author would like to thank Piotr T. Chruściel, Melanie Graf,
Michael Kunzinger, Ettore Minguzzi, and Roland Steinbauer, the organizers of the Non-regular Spacetime Geometry workshop at the Erwin Schr\"{o}dinger Institute, for the invitation to participate and speak at this wonderful workshop. While at the workshop the author was reminded of the desire to explore a VADB type theorem for the null distance and this paper would not have happened without it.

\section{Background}\label{sect-Background}

A spacetime is a time oriented Lorentzian manifold $(M,g)$. We define $J^{\pm}_g(p)$ to be the set of all points which are in the causal future or past of $p$, with respect to $g$. Let $\tau:M \rightarrow \R$ be a generalized time function which is a function which is strictly increasing along all future directed causal curves. If $\tau$ is continuous then we say that $\tau$ is a time function. Any past distinguishing spacetime (If $I^-(p)=I^-(q)$ then $p=q$) admits a generalized time function \cite{SV} and any stably causal spacetime admits a time function \cite{Bernal-Sanchez} where a stably causal spacetime is one in which no closed causal curves exist even under small pertubations of the spacetime. From here on we will only study spacetimes $M$ which come equipped with a (generalized) time function $\tau$. 

We now define the null length and null distance introduced by Sormani and Vega \cite{SV}. Let $\beta :[a,b] \rightarrow M$ be a piecewise causal curve, i.e. a piecewise smooth curve that is either future-directed or past-directed causal on its pieces $a=s_0 < s_1 < \ldots < s_k =b$ (see figure \ref{fig:AdmissableCurves}). The null length of $\beta$ is given by
 \begin{align}
  \hat L_{\tau,g} (\beta) = \sum_{i=1}^k |\tau(\beta(s_i))-\tau(\beta(s_{i-1}))|.
 \end{align}
 In the case where $\tau$ is differentiable we can compute the null length of $\beta$ by
 \begin{align}
  \hat L_{\tau,g} (\beta) = \int_a^b |(\tau \circ \beta)'(s)| ds.
 \end{align}
 For any $p,q \in M$, the \emph{null distance} is given by
 \begin{align}
  \hat d_{\tau,g} (p,q) = \inf \{ \hat L_{\tau,g} (\beta) : \beta \text{ is a piecewise causal curve from } p \text{ to } q \}.  
 \end{align}
When the time function is clear from context we will write $\hat{d}_{g}$. One should notice that the null distance can be defined on very weak notions of spacetimes such as the Lorentzian length spaces of  M. Kunzinger and C. S\"{a}mann \cite{KCS}. In particular, the null distance can be studied on generalized cones which has been initiated by M. Kunzinger and R. Steinbauer \cite{KS}.

 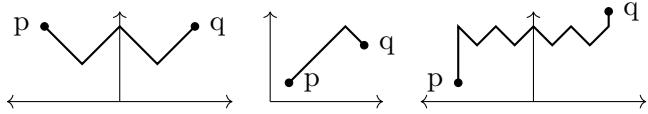
\begin{figure} [h]
  \centering
 \begin{tikzpicture}[scale=1]
  \draw[<->] (-1.5,0) -- (1.5,0) node[anchor=north west]{} ;
  \draw[->] (0,0) -- (0,1.2) node[anchor=south east]{};
  \draw[->] (2,0) -- (3.5,0) node[anchor=north west]{};
  \draw[->] (2,0) -- (2,1.2) node[anchor=south east]{};
    \draw[<->] (4,0) -- (7,0) node[anchor=north west]{};
  \draw[->] (5.5,0) -- (5.5,1.2) node[anchor=south east]{};
  \draw[fill] (-1,1) circle [radius=0.05];
  \node[left, outer sep=2pt] at (-1,1) {p};
  \draw[fill] (1,1) circle [radius=0.05];
  \node[right, outer sep=2pt] at (1,1) {q};
   \draw[ thick] (-1,1) -- (-1/2,1/2) -- (0,1) -- (1/2,1/2) -- (1,1);
   \draw[thick](2+.25,.25)--(2+1,1)-- (2+1.25,.75) ;
    \draw[fill] (2+.25,.25) circle [radius=0.05];
   \node[right, outer sep=2pt] at (2+.25,.25) {p};
  \draw[fill] (2+1.25,.75) circle [radius=0.05];
  \node[right, outer sep=2pt] at (2+1.25,.75) {q};
  \draw[thick](4.5,.25)--(4.5,1)-- (5.5+-3/4,3/4) -- (5.5+-1/2,1) -- (5.5+-1/4,3/4) -- (5.5+0,1)--(5.5+1/4,3/4) -- (5.5+1/2,1) -- (5.5+3/4,3/4)--(6.5,1)--(6.5,1.2); 
    \draw[fill] (4.5,.25) circle [radius=0.05];
   \node[left, outer sep=2pt] at (4.5,.25) {p};
  \draw[fill] (6.5,1.2) circle [radius=0.05];
  \node[right, outer sep=2pt] at (6.5,1.2) {q};
 \end{tikzpicture}
  \caption{Several examples of admissible curves in Minkowski space. When estimating the null distance in warped product spacetimes it will be helpful to construct similar curves.}
  \label{fig:AdmissableCurves}
 \end{figure}

 \subsection{Properties of the Null Distance}

 In this section we review some important properties of the null distance which will be used in this paper. For further properties of null distance see the work by the author and A. Burtscher \cite{AB}, A. Burtscher and L. Garc\'{i}a-Heveling \cite{BG,BG2}, M. Graff and C. Sormani \cite{GS}, M. Kunzinger and R. Steinbauer \cite{KS}, A. Sakovich and C. Sormani \cite{SS}, C. Sormani and C. Vega \cite{SV}, and C. Vega \cite{V}.

 \begin{defn}\label{def-AntiLipschitz}
 Let $(M,g)$ be a spacetime and $f: M \to \R$. We say that $f$ is \emph{anti-Lipschitz} if there exists a distance function $d$ which agrees with the manifold topology so that for any $p,q\in M$ we have
 \begin{align}
 p \le q \Longrightarrow f(q) - f(p) \ge d(p,q).
 \end{align}
\end{defn}

One can observe that the time functions studied in this paper are anti-Lipschitz by Theorem 4.18 of \cite{SV}. Now we recall a foundational result for the null distance given by Sormani and Vega.

\begin{thm}[Theorem 4.6 of \cite{SV}]\label{thm-BasicNullProperties}
 Let $M$ be a spacetime and $\tau$ a time function on $M$. If $\tau$ is locally anti-Lipschitz, then the induced null distance $\hat d_\tau$ is a definite and conformally invariant metric on $M$ that induces the manifold topology.
\end{thm}

In Theorem 3.25 of \cite{SV}, C. Sormani and C. Vega also directly check that warped products (defined on $I \times \Sigma$ where $(\Sigma,\sigma)$ is a complete Riemannian manifold) with differentiable and strictly increasing time functions induce definite null distance metrics which encode causality. It was verified by the author and Burtscher that $(M,\hat{d}_{\tau})$ is a length space in the sense of Burago, Burago, Ivanov \cite{BBI} with counterexamples given in the case where $\tau$ is not continuous or $M$ is incomplete. We also note that the conformal invariance of the null distance can be exploited for the sake of estimates and it is just this technique which is used to prove Theorem \ref{thm-MainTheorem tau}. This argument requires one to choose a time function for the sequence which is not the $t$-coordinate. Hence, in the proof of Theorem \ref{thm-MainTheorem t} we need to estimate the relationship between a sequence of null distance spacetimes with different time functions. Part of the argument uses a relationship between null distances with different time functions established by the author and A. Burtscher \cite{AB}.

 \begin{lem}[Lemma 2.1 of \cite{AB}]\label{lem:timeequiv}
 Let $I$ be a closed interval, $(\Sigma,\sigma)$ be a compact Riemannian manifold and $(I \times \Sigma, g=-dt^2+f(t)^2\sigma)$ a warped product spacetime. If $\tau$ is a time function such that $\tau(t,x) = \phi(t)$ with $\phi'>0$ then there exists a constant $C>0$ such that
 \[
  \frac{1}{C} \hat d_{t,g}(p,q) \leq \hat d_{\tau,g} (p,q) \leq C \hat d_{t,g}(p,q), \qquad p,q \in M.
 \]
\end{lem}

\subsection{Cosmological time functions}

 Given the freedom to choose a time function in the definition of the null distance it is natural to ask what a canonical choice of time function should be? C. Sormani and C. Vega \cite{SV} suggested, and A. Sakovich and C. Sormani \cite{SS}, and A. Burtscher and L. Garc\'{i}a-Heveling \cite{BG2} recently showed, that the cosmological time function is a natural choice of time function when defining the null distance on a spacetime. Anderson, Galloway, and Howard \cite{AGH} defined the cosmological time function of a Lorentzian manifold $(M,g)$ to be
 \begin{align}
     \tau_{AGH}(p)=\sup\{L_g(\gamma):\text{ future timelike } \gamma:[0,1]\rightarrow N,  \gamma(1)=p\},
 \end{align}
 where the Lorentzian length is defined by
 \begin{align}
     L_g(\gamma)=\int_0^1\sqrt{|g(\gamma',\gamma')|}ds.
 \end{align}
Anderson, Galloway, and Howard say that the time function is regular if $\tau_{AGH}<\infty$ and $\tau_{AGH} \rightarrow 0$ along every past inextensible causal curve. We now show that for every spacetime in this paper the $t$-coordinate will be the cosmological time function. This is a well known fact in the case of warped products but we provide a simple proof here for a slightly more general class of spacetimes.
 
 \begin{lem}\label{lem-tIsCTF}
   Let $(M=[t_0,t_1]\times \Sigma,g=-dt^2+f(t,x)^2\sigma)$ be a Lorentzian manifold where $(\Sigma,\sigma)$ is a compact, connected Riemannian manifold, $x \in \Sigma$, $f: [t_0,t_1]\times \Sigma \rightarrow (0,\infty)$, then $t$ is the cosmological time function.  
 \end{lem} 
 \begin{proof}
     Consider a curve $\gamma:[t_0,t_1]\rightarrow M$, $t_0<t_1$, $\gamma(s)=(s,\alpha(s))$ and note that it will be timelike if
 \begin{align}
     g(\gamma',\gamma')< 0 \quad \Rightarrow \quad-1+f^2|\alpha'|_{\sigma}^2<0 \quad \Rightarrow \quad |\alpha'|_{\sigma}^2 < \frac{1}{f^2}.
 \end{align}

 We note that for any other timelike curve $\bar{\gamma}:[\bar{s}_0,\bar{s}_1] \rightarrow M$, $\bar{\gamma}(\bar{s})=(h(\bar{s}),\bar{\alpha}(\bar{s}))$ we know $h'\not =0$ and hence we can re-scale time so that $s(\bar{s})=\int_{\bar{s}_0}^{\bar{s}} h'(\tau)d \tau$ and $\gamma(s)=(s,\alpha(s))$ where $\alpha(s)=\bar{\alpha}(s^{-1}(s))$.

 If we calculate the length of $\gamma$ we find
 \begin{align}
     L_g(\gamma)=\int_{t_0}^{t_1} \sqrt{1-f^2|\alpha'|_{\sigma}^2}ds \le t_1-t_0.
 \end{align}
 In particular, we see that the supremum in the definition of the cosmological time function is achieved by $\gamma(s)=(s,p_{\Sigma})$ and hence $t$ is the cosmological time function. 
 \end{proof} 
 
 This shows that the cosmological time function for every spacetime considered in this paper is the $t$-coordinate. Since the null distance is invariant under shifts in time, i.e. $\hat{d}_{\tau_1,g}=\hat{d}_{\tau_2,g}$ if $\tau_1=\tau_2+C$ (See Lemma 3.17 in \cite{SV}), we see that $t-t_0$ is a regular cosmological time function. See \cite{AGH} for properties of spacetimes with regular cosmological time functions as well as examples where the cosmological time function is not the $t$-coordinate.

\subsection{Null distance on warped products}

When considering a product $M=[t_0,t_1]\times \Sigma$ and a point $p \in M$ we will denote the time coordinate of $p$ by $t(p)$ and the spatial coordinate of $p$ by $p_{\Sigma}$ throughout the paper. If $(\Sigma,\sigma)$ is a Riemannian manifold then the associated distance function will be denoted $d_{\sigma}$. It is useful then to have an explicit formula for the null distance of Minkowski space, first given by C. Sormani and C. Vega \cite{SV} and generalized to Minkowski products by the author and A. Burtscher \cite{AB}.

\begin{lem}[Proposition 3.3 of \cite{SV}, Lemma 4.4 of \cite{AB}]\label{lem:nulld1}
 Let $(\Sigma,\sigma)$ be a connected continuous Riemannian manifold and the Lorentzian product $(M=[t_0,t_1]\times \Sigma,g=-dt^2+\sigma)$. Then the null distance $\hat d_{t,g}$ satisfies
 \begin{align}\label{d:genMink}
  \hat d_{t,g} (p,q) = 
  \begin{cases}
                           |t(p) - t(q)| & q \in J^\pm_g(p), \\
                           d_\sigma(p_\Sigma,q_\Sigma) & q \not\in J^\pm_g(p).
\end{cases}
 \end{align}

\end{lem}

Now that we have the explicit expression for the null distance on Lorentzian products the author and Burtscher proved bi-Lipschitz bounds for warped products with bounded warping functions. The analagous result was shown for generalized cones by Kunzinger and Steinbaeur \cite{KS}.

\begin{thm}[Proposition 4.10 of \cite{AB}]\label{thm:biLip}
 Let $(\Sigma,\sigma)$ be a connected, continuous Riemannian manifold. Let $f$ be bounded on  $[t_0,t_1]$  such that
\begin{align}
   0< f_{\min} \leq f(t) \leq f_{\max} < \infty, \qquad t\in [t_0,t_1].
  \end{align}
  Consider $M=[t_0,t_1]\times \Sigma$ and the Lorentzian metrics $g_f=-dt^2+f(t)^2\sigma$ and $g=-dt^2+\sigma$.
 Then for all $p,q \in M$
 \begin{align}\label{dfestimate}
 \min \{1, f_{\min} \} \, \hat d_{t,g} (p,q) \leq \hat d_{t,g_f} (p,q) \leq  \max \{1,
  f_{\max} \} \, \hat d_{t,g} (p,q).
 \end{align}
\end{thm}

The author and Burtscher then showed that warped product spacetimes, and more generally globally hyperbolic spacetimes, are integral current spaces using these bi-Lipschitz estimates. In addition, they used these bi-Lipschitz estimates for a sequence of warped products with uniformly bounded warping functions to show that a subsequence of the corresponding null distances must converge by the Arzel\'{a}-Ascoli theorem. This type of argument will also by used in the proof of Theorem \ref{thm-MainTheorem t} of this paper.

\subsection{Uniform Convergence of Metric Spaces}
Given two metric spaces $(X,d_1)$ and $(X,d_2)$, defined on the same set $X$, we can define the uniform distance between them to be
\begin{align}\label{def-UniformDist}
    d_{unif}(d_0,d_1)=\sup_{x_1,x_2\in X}|d_1(x_1,x_2)-d_2(x_1,x_2)|
\end{align}
Given a sequence of metric spaces $(X,d_j)$ we can define the uniform convergence of $d_j$ to a limiting metric space $(X,d_{\infty})$ by
\begin{align}\label{def-UniformConvergence}
    d_{unif}(d_j,d_{\infty}) \rightarrow 0.
\end{align}

As a sanity check for the null distance and convergence, the author and A. Burtscher showed that if one assumed uniform convergence of a sequence of warping functions then one can conclude uniform convergence for the corresponding sequence of null distance functions. In the Riemannian case this is more or less immediate but it is not quite so straight forward in the null distance case since there isn't as direct a link between the warping function and the null lengths of curves as in the Riemannian case. Nonetheless, we were able to establish this expected result. It is the goal of this paper in part to give the optimal version of this theorem for sequences of warped products.

\begin{thm}[Theorem 5.1 of \cite{AB}]\label{thm:fuconv}
 Let $I$ be a closed interval and $(\Sigma,\sigma)$ be a connected compact Riemannian manifold. Suppose $f_j$ is a sequence of continuous functions $f_j \colon I \to (0,\infty)$ (uniformly bounded away from $0$) and $M_j = I \times_{f_j} \Sigma$ are warped products with Lorentzian metric tensors
 \[
  g_j = -dt^2 + f_j(t)^2 \sigma.
 \]
 Assume that $f_j$ converges uniformly to a limit function
 \[ f_\infty(t) = \lim_{j\to\infty} f_j(t). \]
 Then the corresponding null distances $\hat d_j$ of $g_j$, $j \in \N$, with respect to the canonical time function $\tau(t,x)=t$, converge uniformly to $\hat d_\infty$ on $M$.
\end{thm}

The author and Burtscher also give the first example of a sequence of spacetimes whose Gromov-Hausdorff limits and Sormani-Wenger intrinsic flat limits disagree. Although this is not our focus here, Example \ref{Ex-BlowUpSpline} will give another example where the Gromov-Hausdorff limits and Sormani-Wenger intrinsic flat limits will disagree since the generalized spline in the limit will be of a lower dimension than the sequence.

\subsection{Gromov-Hausdorff Convergence}

In this section we recall a definition of the Gromov-Hausdorff distance between metric spaces. We note that this is not the usual first definition given but it seems like the right definition to give to understand this paper. The reader should consult chapter 7 of D. Burago, Y. Burago, S. Ivanov \cite{BBI} for a more comprehensive discussion of Gromov-Hausdorff distance.

\begin{defn}
    Let $X_1, X_2$ be two sets and define a \textbf{correspondence} between $X_1$ and $X_2$ to be a $\mathcal{R}\subset X_1 \times X_2$ so that: $\forall x_1 \in X_1$, $\exists x_2 \in X_2$ so that $(x_1,x_2) \in \mathcal{R}$ and $\forall x_2 \in X_2$, $\exists x_1 \in X_1$ so that $(x_1,x_2) \in \mathcal{R}$.
\end{defn}

A common way to define a correspondence between two sets is via a surjective map $f:X_1\rightarrow X_2$ so that
\begin{align}
    \mathcal{R}=\{(x_1,f(x_1)): x_1 \in X_1\}.
\end{align}

We now show how to define the distortion of a correspondence which will be used to quantify the distance between two metric spaces.

\begin{defn}
   Let $\mathcal{R}$ be a correspondence between two metric spaces $(X_1,d_1)$,  $(X_2,d_2)$. The \textbf{distortion} of $\mathcal{R}$ is
   \begin{align}
       \dis\mathcal{R}=\sup\{|d_1(x_1,x_1')-d_2(x_2,x_2')|: (x_1,x_2),(x_1',x_2') \in \mathcal{R}\}.
   \end{align}
\end{defn}

We can now give a definition of the Gromov-Hausdorff distance between two metric spaces.

\begin{defn}\label{def-GHdistance}
  The \textbf{Gromov-Hausdorff distance} between two metric spaces $(X_1,d_1)$,  $(X_2,d_2)$ is
   \begin{align}
  d_{GH}((X_1,d_1),(X_2,d_2))=\frac{1}{2} \inf_{\mathcal{R}}\dis \mathcal{R},
   \end{align}
   where the infimum is taken over all correspondences between $X_1$ and $X_2$.
\end{defn}

Given a notion of distance between spaces we can now define the corresponding notion of convergence.

\begin{defn}\label{def-GHconvergence}
   Given a sequence of metric spaces $(X_i,d_i)$ we say that the sequence converges to the limit metric space $(X_{\infty},d_{\infty})$ in the Gromov-Hausdorff sense if as $i \rightarrow \infty$ we find
   \begin{align}
       d_{GH}((X_i,d_i),(X_{\infty},d_{\infty})) \rightarrow 0.
   \end{align}
\end{defn}

We note that when the set defining the metric space is the same, uniform convergence of a sequence of metric spaces implies Gromov-Hausdorff convergence of the sequence.  In particular, the quantity defined in \eqref{def-UniformDist} which is going to $0$ in \eqref{def-UniformConvergence} is showing that the distortion of Definition \eqref{def-GHdistance} is going to $0$. In this paper the sequences we consider will mostly be defined on the same set $X=[t_0,t_1]\times \Sigma$ and hence uniform convergence will almost always be equivalent to Gromov-Hausdorff convergence. The only two exceptions to this rule show up in example \ref{Ex-BlowUpBubble} and example \ref{Ex-BlowUpSpline} where the limiting metric space is defined on a different set. 
We now want to introduce a useful way of establishing the Gromov-Hausdorff convergence of a sequence of metric spaces. This will be used in example \ref{Ex-BlowUpBubble} and example \ref{Ex-BlowUpSpline} where the topology of the limit changes and hence uniform convergence does not apply.

\begin{defn}\label{def-AlmostIsometry}
   Let $\varepsilon>0$ and $(X_1,d_1)$,  $(X_2,d_2)$ be two metric spaces. A function $f:X_1\rightarrow X_2$ is called an \textbf{$\varepsilon$-isometry} if it is almost distance preserving
   \begin{align}
      \dis \{(x_1,f(x_1)):x_1\in X_1\} \le \varepsilon, 
   \end{align} and it is almost onto 
   \begin{align}
     \sup\{d_2(f(X_1),x_2):x_2 \in X_2\}\le \varepsilon.  
   \end{align}
\end{defn}

We now give a useful theorem for showing Gromov-Hausdorff convergence where we note that the function in the theorem need not be continuous and will not be when used in example \ref{Ex-BlowUpBubble} and example \ref{Ex-BlowUpSpline}.

\begin{thm}\label{thm-AlmostIsometryImpliesGH}
    Let $(X_j,d_j)$ be a sequence of metric spaces and $(X_{\infty},d_{\infty})$ a desired limiting metric space.  If for all $j \in \N$ there exists a $\frac{1}{j}-isometry$ $f:X_j\rightarrow X_{\infty}$ then $(X_j,d_j)$ converges to $(X_{\infty},d_{\infty})$ in the Gromov-Hausdorff sense.
\end{thm}

We also note that a Gromov-Hausdorff limit of Alexandrov spaces converges to a metric space which is also an Alexandrov space (See Proposition 10.7.1 \cite{BBI}). Hence, there is already a synthetic notion of sectional curvature defined for spacetimes (Lorentzian length spaces) with the null distance. It would be interesting to explore further the usefulness of this definition in the case of spacetimes and Lorentzian length spaces. It should also be noted that synthetics curvature bounds have been defined on Lorentzian length spaces by M. Kunzinger and C. S\"{a}mann \cite{KCS} and have been previously defined on Lorentzian manifolds by S. Harris \cite{SH}. It could be interesting to compare the Alexandrov spaces one obtains from the null distance to the curvature bounds defined in \cite{KCS, SH}.

\section{Examples}\label{sect-Examples}

In this section we have several goals: To demonstrate the importance of the hypotheses in Theorem \ref{thm-MainTheorem t} and to point out the difficulty in extending this theorem to more general globally hyperbolic spacetimes.  Furthermore, the examples in this section are crucial to understanding the convergence of spacetimes with the null distance under Gromov-Hausdorff and spacetime intrinsic flat convergence, but should also be important for testing the various notions of convergence of spacetimes which have been explored recently \cite{CM,MS,M,SSFuture}. 

Throughout this section $(\mathbb{D}^n,\sigma)$ stands for the closed flat unit disk which is chosen for notational convenience. It should be noted that the conclusions of the examples does not rely heavily on this choice and similar examples will hold for a compact Riemannian manifold $(\Sigma^n,\sigma)$ replacing $\mathbb{D}^n$.

\subsection{Shortcut Example}

First we remember an example given by the author and A. Burtscher \cite{AB} of a sequence of warped products $f_j$ converging in $L^1$ to $f_{\infty}$, where $f_j \not \ge (1-\frac{1}{j})f_{\infty}$, and so that $\hat{d}_{g_j}$ does not converge uniformly to $\hat{d}_{g_{\infty}}$. This example shows the necessity of $C^0$ convergence from below of the warping functions in Theorem \ref{thm-MainTheorem tau} and Theorem \ref{thm-MainTheorem t}. One should note that a similar observation was made by the author and C. Sormani \cite{Allen-Sormani, Allen-Sormani-2} in the Riemannian case.

	\begin{ex}[Example 5.7 of \cite{AB}]\label{Ex-Shortcut}
	For any fixed compact Riemannian manifold $(\Sigma,\sigma)$ and $0<h_0<1$, the sequence of warping functions $f_j$, depicted below, defines a sequence of warped product Lorentzian metrics \begin{align}
	   g_j=-dt^2+f_j^2\sigma
	\end{align} 
 with induced null distances $\hat{d_j}$ on the manifold\ $M = [0,2] \times \Sigma$.
	
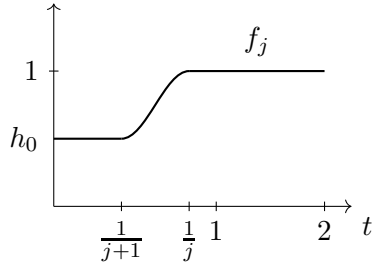
\begin{center}
 \begin{tikzpicture}[scale=1.8]
  \draw[->] (0,0) -- (2.2,0) node[anchor=north west] {$t$};
  \draw[->] (0,0) -- (0,1.5) node[anchor=south east] {};
  
  \foreach \x/\xtext in { 0.5/\frac{1}{j+1},1/\frac{1}{j},  1.2/1, 2/2}
    \draw[shift={(\x,0)}] (0pt,1pt) -- (0pt,-1pt) node[below] {$\xtext$};

  \draw[thick] (1,1) -- (2,1);
  \node[above, outer sep=2pt] at (1.5,1) {$f_j$};
  \draw[thick] (0,1/2) -- (1/2,1/2);
  \node[left, outer sep=2pt] at (0,.5) {$h_0$};
  \draw[thick] (1/2,1/2) cos (3/4,3/4) sin (1,1);

  \node[left, outer sep=2pt] at (0,1) {$1$};
  \draw[shift={(0,1)}] (-1pt,0pt) -- (1pt,0pt);
 \end{tikzpicture}
 \end{center}

 The sequence of null distances $\hat{ d}_{t,g_j}$ induced by the warping functions $f_j$ converges uniformly to the metric $d_0 \neq  \hat {d}_{t,g}$, $g=-dt^2+\sigma$.
 \vspace{0.5cm}

 For points $p=(t(p),p_\Sigma),q=(t(q),q_\Sigma) \in [0,2] \times \Sigma$  the metric $d_0$ is given by
 \begin{align}
  d_0(p,q) = \min\left\{\hat{ d}_\sigma(p,q), t(p)+t(q) + h_0 \, \hat{d}_\sigma(J_{\sigma}^-(p)\cap  \Sigma_0,J_{\sigma}^-(q)\cap  \Sigma_0)\right\}
 \end{align}
 where $\Sigma_{t} = \{t\}\times \Sigma$.
 \end{ex}

\subsection{Blow up which does not effect the sequence}

Next we generalize an example given by the author and Burtscher \cite{AB} where the following sequence of warping functions, given in Figure \ref{fig-ShortcutFigure}, was considered.
\begin{figure} [h]\label{fig-ShortcutFigure}
  \centering
  \begin{tikzpicture}[scale=2]
  \draw[->] (0,0) -- (2.2,0) node[anchor=north west] {$t$};
  \draw[->] (0,0) -- (0,2) node[anchor=south east] {};
  
  \foreach \x/\xtext in { 0.7/\frac{1}{j+1}, 1/\frac{1}{j}, 1.2/1, 2/2}
    \draw[shift={(\x,0)}] (0pt,1pt) -- (0pt,-1pt) node[below] {$\xtext$};

  \draw[thick] (1,1) -- (2,1);
  \node[above, outer sep=2pt] at (1.5,1) {$f_j$};
  \draw[thick] (0,1+1/2) -- (1/2,1+1/2);
  \node[left, outer sep=2pt] at (0,1.5) {$h_0$};
    \draw[shift={(0,1)}] (-1pt,0pt) -- (1pt,0pt);
    \node[left, outer sep=2pt] at (0,1) {$1$};
  \draw[thick] (1/2,1+1/2) cos (3/4,1+1/4) sin (1,1);
  \end{tikzpicture}
 \end{figure} 
It was observed that unlike Example \ref{Ex-Shortcut}, when the warping function converges pointwise a.e. from above the sequence of null distance warped product spacetimes converges to a Minkowski product with the null distance. In the next example we generalize this by allowing the metrics to blow up along the sequence where we see that the example is completely agnostic to the blow up rate of the sequence.

  \begin{ex}\label{Ex:BlowUpNoEffect}
      Let $f_j:[0,1]\times \mathbb{D}^n$ be a sequence of continuous functions defined by
      \begin{align}f_j(t)
          \begin{cases}
          j^{\eta} & t \in [0,\frac{1}{2j}]
          \\h_j & t \in \ [\frac{1}{2j},\frac{1}{j}])
\\1& \text{ otherwise }
          \end{cases}
      \end{align}
      where $\eta \in (0,\infty)$, and $h_j$ is any continuous decreasing function so that $h_j(\frac{1}{2j})= j^{\eta}$ and $h_j(\frac{1}{j})=1$. If $g_j=-dt^2+f_j^2\sigma$, $g_0=-dt^2+\sigma$ on $[0,1]\times \mathbb{D}$ then we find
      \begin{align}
          \hat{d}_{t,g_j} \rightarrow \hat{d}_{t,g_0},
      \end{align}
      uniformly.
  \end{ex}
  \begin{proof}
      Consider $p,q \in  [0,1]\times\mathbb{D}^2 $. If $q \in J^{\pm}_{g_j}(p)$ then since the null cones with respect to $g_0$ are wider than that of $g_j$ we see that $\hat{d}_{t,g_j}(p,q)=\hat{d}_{t,g_0}(p,q)$ and there is nothing to show. 
      
      Additionally, if $p,q \in [\frac{1}{j},1] \times \mathbb{D}^2 $ then it is also the case that $\hat{d}_{t,g_j}(p,q)=\hat{d}_{t,g_0}(p,q)$ and there is nothing to show.

      Hence, assume that $p \in [0,\frac{1}{j}]\times \mathbb{D}^2$ and $q \not \in J^{\pm}_{g_j}(p)$. Now define $p'=(p_{\mathbb{D}^2},\frac{1}{j})$ and $q'=(q_{\mathbb{D}^2},\frac{1}{j})$ so that $\hat{d}_{g_j}(p,p')\le\frac{1}{j}$, $\hat{d}_{g_j}(q,q')\le\frac{1}{j}$. 

      If $q \in [0,\frac{1}{j}]\times \mathbb{D}^2$ then  we see that
      \begin{align}
        \hat{d}_{t,g_j}(p,q)&\le \hat{d}_{t,g_j}(p,p')+\hat{d}_{t,g_j}(p',q')+\hat{d}_{t,g_j}(q',q) 
        \\&\le \frac{2}{j} + d_{\sigma}(p'_{\mathbb{D}^2},q'_{\mathbb{D}^2}) =\frac{2}{j} + d_{\sigma}(p_{\mathbb{D}^2},q_{\mathbb{D}^2}).          
        \end{align}
Then if $q \in J_{g_0}^{\pm}(p)$ we find
\begin{align}
        \hat{d}_{t,g_j}(p,q)&\le \frac{2}{j} + d_{\sigma}(p_{\mathbb{D}^2},q_{\mathbb{D}^2})
        \\&\le \frac{2}{j} + |t(p)-t(q)| =\frac{2}{j}+ \hat{d}_{t,g_0}(p,q).
        \end{align}
Else if $q \not \in J_{g_0}^{\pm}(p)$ we find
           \begin{align}
        \hat{d}_{t,g_j}(p,q)&\le \frac{2}{j}+ \hat{d}_{t,g_0}(p,q). 
        \end{align}
      
      If $q \in [\frac{1}{j},1]\times \mathbb{D}^2$ then  we see that 
      \begin{align}
        \hat{d}_{t,g_j}(p,q)&\le \hat{d}_{t,g_j}(p,p')+\hat{d}_{t,g_j}(p',q) \le \frac{1}{j} + \hat{d}_{t,g_0}(p',q) .          
        \end{align}
        If $q \not \in J^{\pm}_{g_0}(p')$ then by applying Lemma \ref{lem:nulld1} we find
        \begin{align}
        \hat{d}_{t,g_j}(p,q)& \le \frac{1}{j} + d_{\sigma}(p'_{\mathbb{D}^2},q_{\mathbb{D}^2})
  = \frac{1}{j} + d_{\sigma}(p_{\mathbb{D}^2},q_{\mathbb{D}^2}) = \frac{1}{j} + \hat{d}_{t,g_0}(p,q).
      \end{align}
       If $q \in J^{\pm}_{g_0}(p')$
        \begin{align}
        \hat{d}_{t,g_j}(p,q)& \le \frac{1}{j} + |t(p')-t(q)|
         \\&\le \frac{1}{j} + |t(p')-t(p)|+|t(p)-t(q)|  = \frac{2}{j} +\hat{d}_{t,g_0}(p,q).
      \end{align}
Then since $f_j \ge 1$ we see that the set of piecewise causal curves with respect to $g_j$ is strictly smaller than the corresponding set in the Minkowski product space and hence $\hat{d}_{t,g_j} \ge \hat{d}_{t,g_0}$. So we find that
\begin{align}
\hat{d}_{t,g_0}(p,q) \le    \hat{d}_{t,g_j}(p,q) \le \frac{2}{j} + \hat{d}_{t,g_0}(p,q),
\end{align}
and we are done.      
  \end{proof}

\subsection{Blowing up on a dense set}

In the previous example we saw that the blow up rate of $f_j$ may not effect the convergence of the sequence if it is happening at an isolated point. In the next example we see that if $f_j$ converges pointwise almost everywhere to $1$, but $f_j$ blows up along a dense set of points, then the blow up rate can effect the convergence of the sequence in dramatic fashion. In fact, we will see that the limit of the sequence of null distance metric spaces is a taxi metric on $[0,2]\times \mathbb{D}^n$, further showing the importance of $L^1$ convergence in Theorem \ref{thm-MainTheorem t}.
  
   \begin{ex}\label{Ex-ManyBlowUpsTaxi} 
We construct a sequence of functions
$f_j:[0,1]\to [1,\infty)$ 
by letting
 \begin{align}
 S=\left\{s_{i,j}= \tfrac{i}{2^j}\,: \,  i=1,2,... (2^j-1),\, j\in \mathbb{N}\right\}= \left\{ \tfrac{1}{2},\tfrac{1}{4}, \tfrac{2}{4},  \tfrac{3}{4}, 
\tfrac{1}{8},\tfrac{2}{8},\tfrac{3}{8}...\right\},
 \end{align}
 which is dense in $[0,1]$
 and
 \begin{align}
 \delta_j:=(1/2)^{2j},\quad j \in \mathbb{N} .
 \end{align}
 Let
  \begin{align}
 f_j(t)=
 \begin{cases}
  h_j((t-s_{i,j})/\delta_j ) & t\in [s_{i,j}-\delta_j, s_{i,j} +\delta_j] \textrm{ for } i =1...2^j-1
 \\ 1 & \textrm{ elsewhere }
 \end{cases}
\end{align}
where $h_j:[-1,1]\rightarrow [1,\infty)$ is a continuous function such that 
$h_j(-1)=1$, 
increasing up to $h_j(t)=j^{\eta}2^j$, $\eta \in (0,\infty)$ for $t \in [-\frac{1}{2},\frac{1}{2}]$ and then
decreasing back down to $h_j(1)=1$. 
If we define $g_j=-dt^2+f_j(t)^2 \sigma$ on $M=[0,1]\times \mathbb{D}^n$ and 
\begin{align}
    d_{\infty}(p,q)=|t(p)-t(q)|+d_{\sigma}(p_{\mathbb{D}^n},q_{\mathbb{D}^n}),
\end{align}
then we find that
\begin{align}
    \hat{d}_{t,g_j} \rightarrow d_{\infty},
\end{align}
uniformly.
 \end{ex}

 \begin{proof}
If we define $g_0=-dt^2+\sigma$ and let $p \in M$ then we notice that $J^{\pm}_{g_j}(p) \subset J^{\pm}_{g_0}(p)$. For any $p,q \in M$ we can without loss of generality assume $t(p) \le t(q)$ and find a $t_j \in [t(p)-\delta_j,t(q)+\delta_j]\cap [0,1]$ so that $f_j(t)=1$ for $t \in [t_j-\frac{\delta_j}{2},t_j+\frac{\delta_j}{2}]$. Then by the triangle inequality we know
\begin{align}
    \hat{d}_{t,g_j}(p,q)& \le \hat{d}_{t,g_j}(p,(t_j,p))+\hat{d}_{t,g_j}((t_j,p),(t_j,q))+\hat{d}_{t,g_j}((t_j,q),q)\label{MiddleTermEst}
    \\&\le |t(p)-t_j|+d_{\sigma}(p_{\mathbb{D}^n},q_{\mathbb{D}^n})+|t(q)-t_j|
       \\&\le |t(p)-t(q)|+d_{\sigma}(p_{\mathbb{D}^n},q_{\mathbb{D}^n})+2\delta_j= d_{\infty}(p,q)+2\delta_j,
\end{align}
where the estimate of the middle term in \eqref{MiddleTermEst} is found by a zig-zag piecewise Minkowski null curve in the $[t'-\frac{\delta_j}{2},t'+\frac{\delta_j}{2}]$ region.

     Next we claim that for all $p \in M$  we find  $\bigcap_{j \in \N}J^{\pm}_{g_j}(p)=\{(t,p_{\mathbb{D}^n}) \in M: t \in [0,1]\}$. To this end, if we consider $\beta:(a,b) \rightarrow M$ so that $\beta'(t)=(f_j(t),\alpha'(t))$, $|\alpha'(t)|_{\sigma}=1$, a null curve with respect to $g_j$, then we notice that 
 \begin{align}
     \int_0^1 f_j(t)dt \ge \sum_{i=1}^{2^j-1} \delta_j j^{\eta}2^j=j^{\eta}2^j\frac{2^j-1}{2^{2j}} \rightarrow \infty,
 \end{align}
 and by construction we also see that for any $(a,b) \subset [0,1]$ and $j$ large enough 
\begin{align}\label{NullCurvesBecomeVertical}
     \int_a^b f_j(t)dt \ge \frac{|b-a|}{2} j^{\eta}2^j\frac{2^j-1}{2^{2j}} \rightarrow \infty.
 \end{align}
 So we see that the null curves with respect to $g_j$ are becoming vertical as $j \rightarrow \infty$. This implies that the only causal curves which remain null for the entire sequence are the curves which move solely in the $t$ direction. 

 Let $t_j \in [0,1]$ so that $|t(q)-t_j|\le \delta_j$ and $f_j=1$ on $(t_j-\frac{\delta_j}{2},t_j+\frac{\delta_j}{2})$. Due to the construction of $f_j$ we see that an optimal piecewise null curve with respect to $g_j$ connecting $p$ to $q$ would connect $p$ to a point $p_j \in (t_j-\frac{\delta_j}{2},t_j+\frac{\delta_j}{2})\times \mathbb{D}^n$ by a null curve with respect to $g_j$, take advantage of the $g_0$ null cones in this region, and then travel from $q_j \in (t_j-\frac{\delta_j}{2},t_j+\frac{\delta_j}{2})\times \mathbb{D}^n$ to $q$ along one final null curve with respect to $g_j$. Then by \eqref{NullCurvesBecomeVertical} and the definition of $p_j$ and $q_j$ we find that
 \begin{align}
     \max \{d_{\sigma}((p_j)_{\mathbb{D}^n},p_{\mathbb{D}^n}),d_{\sigma}((q_j)_{\mathbb{D}^n},q_{\mathbb{D}^n})\} \le \frac{C}{j}
 \end{align}
 which implies
 \begin{align}
   d_{\sigma}((p_j)_{\mathbb{D}^n},(q_j)_{\mathbb{D}^n}) &\ge d_{\sigma}(p_{\mathbb{D}^n},q_{\mathbb{D}^n})-d_{\sigma}((p_j)_{\mathbb{D}^n},p_{\mathbb{D}^n})-d_{\sigma}((q_j)_{\mathbb{D}^n},q_{\mathbb{D}^n})
   \\&\ge  d_{\sigma}(p_{\mathbb{D}^n},q_{\mathbb{D}^n})-\frac{2C}{j}.
 \end{align}
Putting all of this together we find
 \begin{align}
   \hat{d}_{t,g_j}(p,q)+\frac{1}{j}&\ge |t(p)-t(p_j)|+d_{\sigma}((p_j)_{\mathbb{D}^n},(q_j)_{\mathbb{D}^n})+|t(q_j)-t(q)|
   \\&\ge |t(p)-t(q)|-2\delta_j+d_{\sigma}(p_{\mathbb{D}^n},q_{\mathbb{D}^n})-\frac{2C}{j},
   \\&\ge d_{\infty}(p,q)-2\delta_j-\frac{2C}{j},
 \end{align}
 and hence we are done.

 \end{proof}

 \subsection{A dense network of shortcuts}

 In the next example we want to see if instead of $L^1$ convergence to $f_{\infty}$ one could just require pointwise convergence along a dense set to $f_{\infty}$. To test this hypothesis we consider an example where $1=f_{\infty} \le f_j$, $f_j$ converges to $f_{\infty}$ pointwise along a dense set of points, but does not converge to $f_{\infty}$ in $L^1$. We will see that $\hat{d}_{t,g_j}$ will not converge to $\hat{d}_{t,g_{\infty}}$ and instead will converge to a metric space which is not the null distance induced by a Lorentzian product, further showing the importance of $L^1$ convergence in Theorem \ref{thm-MainTheorem t}. At first glance the limiting metric space of example \ref{Ex-ManyShortcutTaxi} may seem exactly the same as Lemma \ref{lem:nulld1} but one should notice that the causal future and past is defined with respect to $-dt^2+4\sigma$ in the limiting metric and hence is not the null distance with respect to a Lorentzian product.

  \begin{ex}\label{Ex-ManyShortcutTaxi} 
We construct a sequence of functions
$f_j:[0,1]\to [1,2]$ 
by letting
 \begin{align}
 S=\left\{s_{i,j}= \tfrac{i}{2^j}\,: \,  i=1,2,... (2^j-1),\, j\in \mathbb{N}\right\}= \left\{ \tfrac{1}{2},\tfrac{1}{4}, \tfrac{2}{4},  \tfrac{3}{4}, 
\tfrac{1}{8},\tfrac{2}{8},\tfrac{3}{8} ,...\right\}
 \end{align}
 which is dense in $[0,1]$
 and
 \begin{align}
 \delta_j:=(1/2)^{2j},\quad j \in \mathbb{N}.
 \end{align}
 Let
  \begin{align}
 f_j(t)=
 \begin{cases}
  h((t-s_{i,j})/\delta_j ) & t\in [s_{i,j}-\delta_j, s_{i,j} +\delta_j] \textrm{ for } i =1...2^j-1
 \\ 2 & \textrm{ elsewhere }
 \end{cases}
\end{align}
where $h:[-1,1]\rightarrow [1,2]$ is a continuous function such that 
$h(-1)=2$, 
decreasing down to $h(t)=1$ for $t \in [-\frac{1}{2},\frac{1}{2}]$ and then
increasing back up to $h(1)=2$. 
If we define $g_j=-dt^2+f_j(t)^2 \sigma$, $g=-dt^2+4\sigma$  on $[0,2]\times \mathbb{D}^n$ 
and 
\begin{align}\label{d:genMink}
   d_0(p,q) = 
  \begin{cases}
 |t(p) - t(q)| & q \in J^\pm_g(p), \\
|t(p) - t(q)|+ d_\sigma((J^{\pm}_g(p)\cap (t(q) \times\mathbb{D}^n))_{\mathbb{D}^n},q_{\mathbb{D}^n}) & q \not\in J^\pm_g(p),
\end{cases}
 \end{align}
then we find that
\begin{align}
    \hat{d}_{t,g_j} \rightarrow d_0,
\end{align}
uniformly.
 \end{ex}

  \begin{proof}
  If we define $g_0=-dt^2+\sigma$ and let $p \in M$ then we notice that $J^{\pm}_{g}(p) \subset J^{\pm}_{g_j}(p) \subset J^{\pm}_{g_0}(p)$.
 So for $q \in J^{\pm}_{g}(p)$ we find $\hat{d}_{g_j,t}(p,q)=|t(p)-t(q)|= d_0(p,q)$. 
 
 Now consider $q \not\in J^{\pm}_{g}(p)$. If $t(p)=t(q)=t'$ then there exists a $t_j \in [0,1]$ so that $|t'-t_j|\le \delta_j$ and $f_j=1$ on $(t_j-\frac{\delta_j}{2},t_j+\frac{\delta_j}{2})$ by construction. Hence one can build an admissible piecewise causal curve with respect to $g_j$ by first traveling from $p$ solely in the $t$-direction into the region $(t_j-\frac{\delta_j}{2},t_j+\frac{\delta_j}{2}) \times \mathbb{D}^n$, then traveling along a piecewise Minkowski null curve contained in $(t_j-\frac{\delta_j}{2},t_j+\frac{\delta_j}{2}) \times \mathbb{D}^n$, and back along a curve traveling solely in the $t$-direction to $q$ so that 
 \begin{align}
     \hat{d}_{t,g_j}(p,q) \le d_{\sigma}(p_{\mathbb{D}^n},q_{\mathbb{D}^n})+\frac{\delta_j}{2}.
 \end{align}
 Since $J^{\pm}_{g_j}(p) \subset J^{\pm}_{g_0}(p)$ we know that 
 \begin{align}
  d_{\sigma}(p_{\mathbb{D}^n},q_{\mathbb{D}^n})=   \hat{d}_{t,g_0}(p,q) \le   \hat{d}_{t,g_j}(p,q),
 \end{align}
 and we are done with this case.

Now assume $t(p)\not = t(q)$ and a $q'\in J^{\pm}_{g}(p)$ so that $t(q')=t(q)$ and $d_{\sigma}(q'_{\mathbb{D}^n},q_{\mathbb{D}^n})$ is minimal. Then we know by the triangle inequality and the previous case that
\begin{align}
    \hat{d}_{t,g_j}(p,q)& \le \hat{d}_{t,g_j}(p,q')+\hat{d}_{t,g_j}(q',q)
    \\&\le|t(p)-t(q)|+d_{\sigma}(q'_{\mathbb{D}^n},q_{\mathbb{D}^n})+\frac{\delta_j}{2}.
\end{align}

Now we aim for the estimate from below. If we consider $\beta:(a,b) \rightarrow M$, $(a,b) \subset [0,1]$ so that $\beta'(t)=(f_j(t),\alpha'(t))$, $|\alpha'(t)|_{\sigma}=1$, a null curve with respect to $g_j$, then we notice that 
 \begin{align}\label{ConvergenceInLenthEx}
     \hat{L}_{g_j}(\beta)=\int_a^b f_j(t)dt \rightarrow 2|b-a|,
 \end{align}
 and hence the length of the null curves of $g_j$ converge to the length of the null curves of $g$ in the limit, as well as the endpoints of the respective null curves. Let $t'=t(q)$ and  $t_j \in [0,1]$ so that $|t'-t_j|\le \delta_j$ and $f_j=1$ on $(t_j-\frac{\delta_j}{2},t_j+\frac{\delta_j}{2})$. Due to the construction of $f_j$ we see that an optimal piecewise null curve with respect to $g_j$ connecting $p$ to $q$ would connect $p$ to a point $p_j \in (t_j-\frac{\delta_j}{2},t_j+\frac{\delta_j}{2})\times \mathbb{D}^n$ by a null curve with respect to $g_j$, take advantage of the $g_0$ null cones in this region, and then travel from $q_j \in(t_j-\frac{\delta_j}{2},t_j+\frac{\delta_j}{2})\times \mathbb{D}^n$ to $q$ along one final null curve with respect to $g_j$. Then by the consequences of \eqref{ConvergenceInLenthEx} and the definition of $p_j$ and $q_j$ there must exist a $p' \in J^{\pm}_{g}(p)$, $t(p')=t(q)$ so that
 \begin{align}
     \max\{d_{\sigma}((p_j)_{\mathbb{D}^n},p'_{\mathbb{D}^n}),d_{\sigma}((q_j)_{\mathbb{D}^n},q_{\mathbb{D}^n})\} \le \frac{C}{j}
 \end{align}
 which implies
 \begin{align}
   d_{\sigma}((p_j)_{\mathbb{D}^n},(q_j)_{\mathbb{D}^n}) &\ge d_{\sigma}(p'_{\mathbb{D}^n},q_{\mathbb{D}^n})-d_{\sigma}((p_j)_{\mathbb{D}^n},p'_{\mathbb{D}^n})-d_{\sigma}((q_j)_{\mathbb{D}^n},q_{\mathbb{D}^n})
   \\&\ge  d_{\sigma}(p'_{\mathbb{D}^n},q_{\mathbb{D}^n})-\frac{2C}{j}.
 \end{align}
Putting all of this together we find
 \begin{align}
   \hat{d}_{t,g_j}(p,q)+\frac{1}{j}&\ge |t(p)-t(p_j)|+d_{\sigma}((p_j)_{\mathbb{D}^n},(q_j)_{\mathbb{D}^n})+|t(q_j)-t(q)|
   \\&\ge |t(p)-t(q)|-2\delta_j+d_{\sigma}(p'_{\mathbb{D}^n},q_{\mathbb{D}^n})-\frac{2C}{j},
     \\&\ge d_0(p,q) -2\delta_j-\frac{2C}{j},
 \end{align}
 and hence we are done.
 \end{proof}

 One should notice that example \ref{Ex-ManyShortcutTaxi} is a member of a family of examples. If one defines $g^k=-dt^2+k^2\sigma$ for $0<k< \infty$ then we can generalize example \ref{Ex-ManyShortcutTaxi} by defining a similar sequence of functions $f_j^k$ so that $1 \le f_j^k \le k$ where $\hat{d}_{g_j^k,t}$ would uniformly converge to 
 \begin{align}\label{d:genMink}
   d_{k-2}(p,q) = 
  \begin{cases}
 |t(p) - t(q)| & q \in J^\pm_{g^k}(p) \\
|t(p) - t(q)|+ d_\sigma((J^{\pm}_{g^k}(p)\cap (t(q) \times\mathbb{D}^n))_{\mathbb{D}^n},q_{\mathbb{D}^n}) & q \not\in J^\pm_{g^k}(p)
\end{cases}.
 \end{align}

 Notice as $k \rightarrow \infty$ that $d_{k-2}$ converges to the taxi metric $d_{\infty}(p,q)$ of example \ref{Ex-ManyBlowUpsTaxi}.

\subsection{Blow up producing a bubble}

In Example \ref{Ex:BlowUpNoEffect} we saw that the convergence was completely agnostic to the blow up rate of the sequence. In the next two examples we will see that this is generally not the case even if $f_{\infty}\le f_j$ and $f_j$ converges uniformly to $f_{\infty}$ away from the singular set $[0,1]\times \{0\}\subset [0,1]\times \mathbb{D}^n$. In particular, in the next example we see that the blow up of the sequence produces bubbling phenomenon in the limit, where although the sequence consists of manifolds the limit is not a manifold. It seems it would be interesting to further study examples like this in the context of Lorentzian length spaces of Kunzinger and S\"{a}mann \cite{KS}.

  \begin{ex}\label{Ex-BlowUpBubble}
      Let $(\mathbb{D}^n,\sigma)$, $n \ge 2$ be a flat disk, $f_j:[0,1]\times \mathbb{D}^n$, $j \ge 2$ a sequence of continuous functions defined radially on $\mathbb{D}^n$ by
      \begin{align}f_j(r)
          \begin{cases}
          j & r\in  \left[0,\frac{1}{j}\right]
          \\h_j(r)& r \in   \left[\frac{1}{j},\frac{3}{2j}\right]
\\1& \text{ otherwise }
          \end{cases}
      \end{align}
      where $h_j$ is any decreasing continuous function so that $h_j(\frac{1}{j})=j$, $h_j(\frac{3}{2j})=1$, and $\displaystyle \int_{\frac{1}{j}}^{\frac{3}{2j}} h_jdr\le \frac{C}{j}$. If $g_j=-dt^2+f_j^2\sigma$ and $g_0=-dt^2+\sigma$ then we find $\hat{d}_{t,g_j} \not\rightarrow \hat{d}_{t,g_0}$. Furthermore, we let $M=([0,1]\times \mathbb{D}^n,\hat{d}_{t,g_0})$, $N_1=([0,1]\times \mathbb{D}^n,\hat{d}_{t,g_0})$, $N_2=([0,1]\times \mathbb{D}^n,\hat{d}_{t,g_0})$, and 
      \begin{align}
         F:[0,1] \times \partial \mathbb{D}^2\subset N_1 \rightarrow  [0,1]\times\{0\}\subset N_2
      \end{align} 
      defined by projection of the second factor.  Then if $(N_1\sqcup N_2, \hat{d}_{t,g_0})$ is the disjoint union of $N_1$ and $N_2$ with the null distance then we can define the metric space $(P,d_0)=(N_1 \sqcup N_2, \hat{d}_{t,g_0}/\sim)$ where we identify points by the map $F$ in order to conclude that    
      \begin{align}
          (M,\hat{d}_{t,g_j}) \rightarrow (P,d_0),
      \end{align}
      in the Gromov-Hausdorff sense.
  \end{ex}
  \begin{proof}
      Here our aim is to apply Definition \ref{def-AlmostIsometry} and Theorem \ref{thm-AlmostIsometryImpliesGH} to show Gromov-Hausdorff convergence. To this end we define a map
      \begin{align}
          \Lambda_j: (P,d_0)\rightarrow (M,\hat{d}_{g_j}),
      \end{align}
      so that
      \begin{align}
          (N_2\setminus ([0,1]\times \{0\})) \subset P &\mapsto M \setminus ([0,1]\times B_{\frac{3}{2j}}(0)) \subset M,\label{Mapdef1}
          \\ N_1 \subset P &\mapsto [0,1]\times B_{\frac{1}{j}}(0) \subset M,\label{Mapdef2}
      \end{align}
      where each of the portions of the maps in \eqref{Mapdef1} and \eqref{Mapdef2} is onto and is defined by radially scaling the second factor. 
      
      We start by checking that $\Lambda_j$ is almost onto. So we need to show that 
      \begin{align}
         \sup\{\hat{d}_{t,g_j}(\Lambda_j(P),x):x \in M\}<\frac{C}{j} 
      \end{align} 
      which we will accomplish by connecting each point $p \in M \setminus \Lambda_j(P)$ to a point $q \in \Lambda_j(P)$ by an admissible curve $\beta_j$ of short length. This will show that
      \begin{align}
         \hat{d}_{t,g_j}(p,q) \le \hat{L}_{g_j}(\beta_j) \le \frac{C}{j}, \quad \forall p \in M \setminus \Lambda_j(P),q \in \Lambda_j(P),
      \end{align} 
      which implies the claim.  To this end, connect a point $p \in M \setminus ([0,1]\times \partial B_{\frac{3}{2j}}(0))$ to a point $q \in [0,1]\times \partial B_{\frac{1}{j}}(0)$ by a null curve $\gamma:\left[\frac{1}{j},\frac{3}{2j}\right] \rightarrow M$ so that $\gamma'(s)=(h_j(s),x)$, $x \in \partial\mathbb{D}^n$ where 
      \begin{align}
          \hat{L}_{t,g_j}(\gamma)=\int_{\frac{1}{j}}^{\frac{3}{2j}} h_jdr\le \frac{C}{j}.
      \end{align} 
      
      Now we need to show that $\Lambda_j$ is  almost distance preserving, i.e. a $\frac{C}{j}$ isometry for some $C>0$. This will be accomplished by studying the distance function connecting different regions of $P$.

If we take $p,q \in  (N_2\setminus ([0,1]\times \{0\})) \subset P $ then we can write $p=(t(p),p_{\mathbb{D}^n})), q=(t(q),q_{\mathbb{D}^n})$ and hence $\Lambda_j(p)=(t(p),p_{\mathbb{D}^n}+\frac{3}{2j}(1-|p_{\mathbb{D}^n}|)p_{\mathbb{D}^n})$, $\Lambda_j(q)=(t(q),q_{\mathbb{D}^n}+\frac{3}{2j}(1-|q_{\mathbb{D}^n}|)q_{\mathbb{D}^n})$. By considering the class of piecwise null curves connecting $p$ to $q$ with respect to $g_0$ which remains outside of $[0,1]\times B_{\frac{3}{2j}}(0))$ we can see that
\begin{align}
\hat{d}_{t,g_0}(p,q)-\frac{C}{j}  \le  \hat{d}_{t,g_j}(\Lambda_j(p),\Lambda_j(q)) \le \hat{d}_{t,g_0}(p,q)+\frac{C}{j}.
\end{align}

If we take $p,q \in  N_1 \subset P $ then we can write $p=(t(p),p_{\mathbb{D}^n})), q=(t(q),q_{\mathbb{D}^n})$ and hence $\Lambda_j(p)=(t(p),\frac{1}{j}p_{\mathbb{D}^n})$, $\Lambda_j(q)=(t(q),\frac{1}{j}q_{\mathbb{D}^n})$. First we assume that $t(p)=t(q)$ and connect these points by a piecewise null curve $\gamma_j:\left[0,\frac{1}{j}\right]\rightarrow M$ so that $\gamma'_j=(\pm j,x)$, $x \in \mathbb{D}^n$ in which case
\begin{align}\label{AlmostIsometryEq1}
    \hat{d}_{t,g_j}(\Lambda_j(p),\Lambda_j(q))\le \hat{L}_{t,g_j}(\gamma) = j\left(\frac{1}{j}|p_{\mathbb{D}^n}-q_{\mathbb{D}^n}|\right)=\hat{d}_{t,g_0}(p,q).
\end{align}
Since it would not be better for $\gamma$ to leave the $[0,1]\times B_{\frac{1}{j}}(0)$ region we see that \eqref{AlmostIsometryEq1} is actually an equality. If $t(p)\not=t(q)$ and $q \in J^{\pm}_{g_0}(p)$ then $\Lambda_j(q) \in J^{\pm}_{g_j}(\Lambda_j(p))$ and the estimate follows since $ \hat{d}_{t,g_j}(\Lambda_j(p),\Lambda_j(q)) =\hat{d}_{t,g_0}(p,q)$. If $q \not\in J^{\pm}_{g_0}(p)$ then one can travel along a null curve of slope $j$ to a $p'$ so that $t(p')=t(q)$ and then connect $p'$ to $q$ as was done in \eqref{AlmostIsometryEq1}, leading to the same estimate.

If we take $p \in  (N_2\setminus ([0,1]\times \{0\})) \subset P $ and $q \in  N_1 \subset P $ then by considering any $p' \in  [0,1]\times \partial B_{\frac{3}{2j}}(0)\subset M$ and $q' \in [0,1]\times \partial B_{\frac{1}{j}}(0)\subset M$ so that $t(p')=t(q')$ we notice by the triangle inequality and the arguments above that
\begin{align}
    \hat{d}_{t,g_j}(\Lambda_j(p),\Lambda_j(q))&\le \hat{d}_{t,g_j}(\Lambda_j(p),p')+\hat{d}_{t,g_j}(p',q')+\hat{d}_{t,g_j}(q',\Lambda_j(q))
    \\&\le \hat{d}_{t,g_0}(p,\Lambda_j^{-1}(q'))+\hat{d}_{t,g_0}(\Lambda_j^{-1}(q'),q)+\frac{C}{j}.\label{QuotientEq}
\end{align}
By the definition of a quotient metric space (Definition 3.1.12 \cite{BBI}) as the infimum over \eqref{QuotientEq} where $\Lambda_j^{-1}(q')$ is in the identified set, we see that 
\begin{align}
    \hat{d}_{t,g_j}(\Lambda_j(p),\Lambda_j(q))\le \hat{d}_{t,g_0}(p,q)+\frac{C}{j}.
\end{align}
Similarly, 
\begin{align}
\hat{d}_{t,g_0}&(p,q)
\\&=\inf_{t \in [0,1],p',q' \in \partial \mathbb{D}^n}\hat{d}_{t,g_0}(p,(t,p'))+\hat{d}_{t,g_0}((t,q'),q)
\\&\le\inf_{t \in [0,1],p',q' \in \partial \mathbb{D}^n} \hat{d}_{t,g_j}(\Lambda_j(p),\Lambda_j((t,p')))+\hat{d}_{t,g_j}(\Lambda_j((t,q')),\Lambda_j(q))   
\\&\le \hat{d}_{t,g_j}(\Lambda_j(p),\Lambda_j(q))+\frac{1}{j},\label{LastEqExBubble}
\end{align}
which gives the opposite inequality. In \eqref{LastEqExBubble} we are using the fact that any almost minimizing $\hat{d}_{t,g_j}$ curve connecting $p$ to $q$ must pass through $\Lambda_j([0,1]\times \partial \mathbb{D}^n)$ where $[0,1]\times \partial \mathbb{D}^n \subset N_1$.

So we see that $\Lambda_j$ is an almost onto, almost isometry and hence $(M,\hat{d}_j)$ is converging in the Gromov-Hausdorff sense to $(P,d_0)$ by Definition \ref{def-AlmostIsometry} and Theorem \ref{thm-AlmostIsometryImpliesGH}.
  \end{proof}

\subsection{Blow up producing a generalized spline}

In the next example we see that for a sequence of functions blowing up at a critical rate the limit will be Minkowski space with a taxi metric defined on $[0,1]\times [0,1]$ attached to the $t$-axis. Since the attached taxi space  is a lower dimensional metric space then the Minkowski space portion it is attached to, we know that the Sormani-Wenger Intrinsic Flat limit will be a subset of the Minkowski space portion of the limit with the restricted metric or possibly the zero space (See a similar argument made in Example 5.11 and discussion in Remark 5.12 of \cite{AB}). Hence this is another example of a sequence of spacetimes with he null distance where the Gromov-Hausdorff limit and the Sormani-Wenger Intrinsic Flat limit disagree. The first such example was given by the author and A. Burtscher in Example 5.11 of \cite{AB}. 

As was noted in \cite{AB}, more technical understanding of integral current spaces and estimates for the Sormani-Wenger Intrinsic Flat distance in the case of null distance spacetimes are needed to precisely state what the limit space is. Nonetheless, we conjecture that the Sormani-Wenger Intrinsic Flat limit is the limit space of Example \ref{Ex-BlowUpSpline} with the taxi space portion removed. In order to avoid the massive expansion of technical background needed for this paper we do not state this fact rigorously in Example \ref{Ex-BlowUpSpline}. We expect that the future work of A. Sakovich and C. Sormani \cite{SSFuture} on Spacetime Intrinsic Flat convergence will be of vital importance to making this type of argument rigorous in the case of null distance spacetimes.

   \begin{ex}\label{Ex-BlowUpSpline}
      Let $(\mathbb{D}^n,\sigma)$, $n \ge 2$ be a flat disk, $f_j:[0,1]\times \mathbb{D}^n$ a sequence of continuous functions defined radially on $\mathbb{D}^n$ by
      \begin{align}f_j(r)
          \begin{cases}
          \frac{j^{\lambda}}{1+\lambda \ln(j)} & r \in \left[0,\frac{1}{j^{\lambda}}\right]
         \\ \frac{1}{r(1-\ln(r))} & r \in \left[\frac{1}{j^{\lambda}},\frac{1}{j}\right]
          \\h_j(r) & r \in  \left[\frac{1}{j},\frac{3}{2j}\right]
\\1& \text{ otherwise }
          \end{cases}
      \end{align}
      where $\lambda >1$ and $h_j$ is any decreasing, continuous function so that $h_j(\frac{1}{j})= \frac{j}{1+\ln(j)}$, $h_j(\frac{3}{2j})=1$, and $\displaystyle \int_{\frac{1}{j}}^{\frac{3}{2j}} h_jdr\le \frac{C}{j}$. If $g_j=-dt^2+f_j^2\sigma$, $g_0=-dt^2+\sigma$ then we find $\hat{d}_{t,g_j}\not\rightarrow \hat{d}_{t,g_0}$. Furthermore, we let $M=([0,1]\times \mathbb{D}^n,\hat{d}_{t,g_0})$, $N=([0,1]\times \mathbb{D}^n,\hat{d}_{t,g_0})$, $L=([0,1]\times [0,1],d_{\text{taxi}}^{\lambda})$, 
      \begin{align}
       d_{\text{taxi}}^{\lambda}((s_1,r_1),(s_2,r_2))=|s_1-s_2|+\lambda|r_1-r_2|,
      \end{align}
      and 
      \begin{align}
         F: [0,1]\times \{1\}\subset L \rightarrow  [0,1] \times \{0\}\subset N 
      \end{align} 
      so that $F(t,1)=(t,0)$. Then if $(N\sqcup L, \bar{d})$ is the disjoint union of $M$ and $N$ we can define the metric space $(P,d_0)=(N \sqcup L, \bar{d}/ \sim)$ where we identify points by the map $F$ in order to conclude that 
      \begin{align}
          (M,\hat{d}_{t,g_j}) \rightarrow (P,d_0),
      \end{align}
      in the Gromov-Hausdorff sense.
  \end{ex}  
  \begin{proof}
      Here our aim is to apply Definition \ref{def-AlmostIsometry} and Theorem \ref{thm-AlmostIsometryImpliesGH} to show Gromov-Hausdorff convergence. To this end we define a map
      \begin{align}
          \Lambda_j: (P,d_0)\rightarrow (M,\hat{d}_{t,g_j}),
      \end{align}
      so that for a fixed $x \in B_{\frac{1}{j}}(0) \setminus B_{\frac{1}{j^{\lambda}}}(0)\subset \mathbb{D}^n$ we define
      \begin{align}
          (N\setminus ([0,1]\times \{0\})) \subset P &\mapsto M \setminus ([0,1]\times B_{\frac{3}{2j}}(0)) \subset M,\label{Map3}
          \\ L\subset P &\mapsto [0,1]\times \left\{a\frac{x}{|x|}: a \in \left[j^{-\lambda},j^{-1}\right]\right\}\subset M,\label{Map4}
      \end{align}
      where the portion of the map defined in \eqref{Map3} and \eqref{Map4} are onto and defined by scaling the second factor. 
      
       We start by checking that this map is almost onto. So we need to show that 
      \begin{align}
         \sup\{\hat{d}_{t,g_j}(\Lambda_j(P),x):x \in M\}<\frac{C}{j} 
      \end{align} 
      which we will accomplish by connecting each point $p \in M \setminus \Lambda_j(P)$ to a point $q \in \Lambda_j(P)$ by an admissible curve $\beta_j$ of short length. This will show that
      \begin{align}
         \hat{d}_{t,g_j}(p,q) \le \hat{L}_{t,g_j}(\beta_j) \le \frac{C}{j}, \quad \forall p \in M \setminus \Lambda_j(P),q \in \Lambda_j(P),
      \end{align} 
      which implies the claim.  To this end, connect a point $p \in M \setminus ([0,1]\times \partial B_{\frac{3}{2j}}(0))$ to a point $q \in [0,1]\times \partial B_{\frac{1}{j}}(0)$ by a null curve $\gamma:\left[\frac{1}{j},\frac{3}{2j}\right] \rightarrow M$ so that $\gamma'(s)=(h_j(s),x)$, $x \in \partial\mathbb{D}^n$ where $\hat{L}_{t,g_j}(\gamma)=\int_{\frac{1}{j}}^{\frac{3}{2j}} h_jdr\le \frac{C}{j}$.  For any $p \in [0,1]\times B_{\frac{1}{j^{\lambda}}}(0)$ we can connect to a point $q \in [0,1]\times B_{\frac{1}{j}}(0)$ where $t(p)=t(q)$ by a piecewise null curve of length less than $\frac{1}{1+\ln(j)}$.
      Lastly, for any point $p \in B_{\frac{1}{j}}(0) \setminus B_{\frac{1}{j^{\lambda}}}(0)$ we can connect to a point $q \in [0,1]\times \left\{a\frac{x}{|x|}: a \in \left[j^{-\lambda},j^{-1}\right]\right\}$ with $t(p)=t(q)$ and $p_{\mathbb{D}^n},q_{\mathbb{D}^n} \in \partial B_r(0)$ for some $r \in [j^{-\lambda},j^{-1}]$ by a curve in $[0,1]\times \partial B_r(0)$ whose length is less than $\frac{1}{1-\ln(r)}$.
      
      Now we need to show that $\Lambda_j$ is almost distance preserving, i.e. a $\frac{C}{j}$ isometry for some $C>0$. This will be accomplished by studying the distance function connecting different regions of $P$.

      If we take $p,q \in  (N\setminus ([0,1]\times \{0\})) \subset P $ then we can write $p=(t(p),p_{\mathbb{D}^2})), q=(t(q),q_{\mathbb{D}^2})$ and hence $\Lambda_j(p)=(t(p),p_{\mathbb{D}^2}+\frac{3}{2j}(1-|p_{\mathbb{D}^2}|)p_{\mathbb{D}^2})$, $\Lambda_j(q)=(t(q),q_{\mathbb{D}^2}+\frac{3}{2j}(1-|q_{\mathbb{D}^2}|)q_{\mathbb{D}^2})$.  By considering the class of piecwise null curves connecting $p$ to $q$ with respect to $g_0$ which remains outside of $[0,1]\times B_{\frac{3}{2j}}(0)$ we can see that
\begin{align}
\hat{d}_{t,g_0}(p,q)-\frac{C}{j}  \le  \hat{d}_{t,g_j}(\Lambda_j(p),\Lambda_j(q)) \le \hat{d}_{t,g_0}(p,q)+\frac{C}{j}.
\end{align}

      If we take $(s_1,r_1),(s_2,r_2) \in L \subset P$ then we can write $\Lambda_j(s_1,r_1)=(s_1,r_1x):=p, \Lambda_j(s_2,r_2)=(s_2,r_2x):=q$ where $x \in B_{\frac{1}{j}}(0) \setminus B_{\frac{1}{j^{\lambda}}}(0)\subset \mathbb{D}^n$ is fixed. If we consider $p'=(s_2,r_1x)$ then we notice that 
      \begin{align}
          \hat{d}_{t,g_j}(p,q)&\le \hat{d}_{t,g_j}(p,p')+\hat{d}_{t,g_j}(p',q)
          \\&=|s_1-s_2|+\hat{d}_{t,g_j}(p',q).
      \end{align}
Now we connect $p'$ to $q$ with a piecewise causal curve 
\begin{align}
    \beta_j(s)=
    \begin{cases}
        (f_j(s|x|), sx) &s \in \left[r_1, \frac{|r_1-r_1|}{2}\right]
        \\(-f_j(s|x|), sx) &s \in \left[\frac{|r_1-r_1|}{2},r_2\right]
    \end{cases}
\end{align}
If we consider the case where $r_1=0$ and $r_2=1$ then we find
\begin{align}
    \hat{L}_{t,g_j}(\beta_j)&= \int_{\frac{1}{j^{\lambda}}}^{\frac{1}{j}} f_j(t)dt
    \\&= \int_{\frac{1}{j^{\lambda}}}^{\frac{1}{j}} \frac{1}{r(1-\ln(r))}dr
        \\&= -\int_{1+\ln(j^{\lambda})}^{1+\ln(j)} u^{-1}du
        \\&= \ln((1+\ln(j^{\lambda}))-\ln((1+\ln(j))= \ln\left(\frac{1+\ln(j^{\lambda})}{1+\ln(j)}\right),
\end{align}
and for any other $0\le r_1<r_2\le 1$ the length of $\beta_j$ will be as above but scaled by $|r_2-r_1|$. So we find that

\begin{align}
          \hat{d}_{t,g_j}(p,q)
          &\le |s_1-s_2|+\hat{d}_{g_j}(p',q)
        \\&  \le |s_1-s_2|+ \hat{L}_{t,g_j}(\beta_j)
        \\&  \le |s_1-s_2|+|r_1-r_2| \left(\ln\left(\frac{1+\ln(j^{\lambda})}{1+\ln(j)}\right)\right)
        \\&  \le |s_1-s_2|+\lambda|r_1-r_2|+\frac{C}{j}=d_0((s_1,r_1),(s_2,r_2))+\frac{C}{j}.
      \end{align}

      Since $\beta_j$ is the most efficient piecewise curve connecting $p'$ and $q$ we see that in fact $\hat{L}_{t,g_j}(\beta_j)=\hat{d}_{g_j}(p',q)$. For the opposite estimate we can first travel from $p$ to $p'$ by a null curve of length $|t(p)-t(p')|$ and then a curve of the form $\beta_j$ and hence
      \begin{align}
          \hat{d}_{t,g_j}(p,q)&= \hat{d}_{t,g_j}(p,p')+\hat{d}_{t,g_j}(p',q)
          \\&\ge |t(p)-t(q)|+\hat{d}_{t,g_j}((t(q),p_{\mathbb{D}^n}),q)-\frac{1}{j}\label{SmallDifferenceEq}
      \end{align}
      where we used the fact that $p,q \in B_{\frac{1}{j}}(0) \setminus B_{\frac{1}{j^{\lambda}}}(0)\subset \mathbb{D}^n$ in \eqref{SmallDifferenceEq}. Now we continue the estimate
       \begin{align}
          \hat{d}_{t,g_j}(p,q)&\ge |s_1-s_2|+\hat{d}_{t,g_j}((t(q),p_{\mathbb{D}^n}),q)-\frac{1}{j}
         \\ &\ge |s_1-s_2|+\hat{L}_{t,g_j}(\beta_j)-\frac{1}{j}
         \\&= |s_1-s_2|+\lambda |r_1-r_2|-\frac{C}{j}=d_0((s_1,r_1),(s_2,r_2))-\frac{C}{j}.
      \end{align}

If we take $p \in  (N\setminus ([0,1]\times \{0\})) \subset P $ and $q \in  L \subset P $ then by considering any $p' \in  [0,1]\times \partial B_{\frac{3}{2j}}(0)$ and $q' \in [0,1] \times \partial B_{\frac{1}{j}}(0)$ so that $t(p')=t(q')$ we notice by the triangle inequality and the arguments above that
\begin{align}
    \hat{d}_{t,g_j}(\Lambda_j(p),\Lambda_j(q))&\le \left(\hat{d}_{t,g_j}(\Lambda_j(p),p')+\hat{d}_{t,g_j}(p',q')+\hat{d}_{t,g_j}(q',\Lambda_j(q))\right)
    \\&\le   \left(\hat{d}_{t,g_0}(p,\Lambda_j^{-1}(q'))+\hat{d}_{t,g_0}(\Lambda_j^{-1}(q'),q)+\frac{C}{j}\right)
    \\&\le  \hat{d}_{t,g_0}(p,\Lambda_j^{-1}(q'))+\hat{d}_{t,g_0}(\Lambda_j^{-1}(q'),q).
\end{align}
By the definition of a quotient metric space (Definition 3.1.12 \cite{BBI}) we see that 
\begin{align}
 \hat{d}_{t,g_j}(\Lambda_j(p),\Lambda_j(q))\le \hat{d}_{t,g_0}(p,q)+\frac{C}{j}.
\end{align}
Similarly, 
\begin{align}
\hat{d}_{t,g_0}(p,q)&=\inf_{t \in [0,1]}\left[\hat{d}_{t,g_0}(p,(t,1))+\hat{d}_{t,g_0}((t,1),q)\right]
\\&\le\inf_{t \in [0,1]}\left[ \hat{d}_{t,g_j}(\Lambda_j(p),\Lambda_j((t,1)))+\hat{d}_{t,g_j}(\Lambda_j((t,1)),\Lambda_j(q))   \right]
\\&\le \hat{d}_{t,g_j}(\Lambda_j(p),\Lambda_j(q))+\frac{1}{j},\label{LastEqExSpline}
\end{align}
which gives the opposite inequality. In \eqref{LastEqExSpline} we are using the fact that any almost minimizing $\hat{d}_{g_j}$ curve connecting $p$ to $q$ must pass through $\Lambda_j([0,1]\times \{1\})$ where $[0,1]\times \{1\} \subset L$.

So we see that $\Lambda_j$ is an almost onto, almost isometry and hence $(M,\hat{d}_{t,g_j})$ is converging in the Gromov-Hausdorff sense to $(P,d_0)$ by Definition \ref{def-AlmostIsometry} and Theorem \ref{thm-AlmostIsometryImpliesGH}.
  \end{proof}

  \subsection{Blow up rate can matter}

  Here we see a further example which shows that the blow up rate is important if the region where the blow up is taking place is not shrinking in the $t$ direction. Unlike the two previous examples though, the blow up rate is below the critical rate and hence we see that the sequence of null distance spacetimes converges to Minkowski space with the null distance. 

  \begin{ex}\label{Ex-BlowupRate}
      Let $(\mathbb{D}^n,\sigma)$ be a flat disk, $f_j:[0,1]\times \mathbb{D}^n$ a sequence of continuous functions defined radially on $\mathbb{D}^n$ by
      \begin{align}f_j(r)
          \begin{cases}
          j^{\alpha} & r \in  \left[0,\frac{1}{2j}\right]
          \\h_j(r) & r \in \left[\frac{1}{2j},\frac{1}{j}\right]
\\1& \text{ otherwise }
          \end{cases}
      \end{align}
      where $\alpha \in (0,1)$ and $h_j$ is any decreasing continuous function so that $h_j(\frac{1}{2j})= j^{\alpha}$ and $h_j(\frac{1}{j})=1$. If $g_j=-dt^2+f_j^2\sigma$, $g_0=-dt^2+\sigma$ then we find
      \begin{align}
          \hat{d}_{t,g_j} \rightarrow \hat{d}_{t,g_0},
      \end{align}
      uniformly.
  \end{ex}
  \begin{proof}
       Consider $p,q \in [0,1]\times \mathbb{D}^n $. Then since $f_j \ge 1$ we see that the set of piecewise causal curves with respect to $g_j$ is strictly smaller than the corresponding set in the Minkowski product space and hence $\hat{d}_{t,g_j} \ge \hat{d}_{t,g_0}$.
       
       If $q \in J^{\pm}(p)$ then $\hat{d}_{t,g_j}(p,q)=\hat{d}_{t,g_0}(p,q)$ and there is nothing to show. 

       If $q \not\in J^{\pm}(p)$ then consider a piecewise null curve $\beta$ which realizes the $\hat{d}_{t,g_0}$ distance between $p$ and $q$. If $\beta$ avoids $[0,1] \times B_{\frac{1}{j}}(0)$ then $\hat{d}_{t,g_j}(p,q)=\hat{d}_{t,g_0}(p,q)$. If not, then let $p'_j,q'_j \in [0,1] \times \partial B_{\frac{1}{j}}(0)$ where $\beta$ intersects $[0,1] \times \partial B_{\frac{1}{j}}(0)$. We also define $p''_j=((p'_j)_{\mathbb{D}^n},t(q'_j))$ which we can connect to $q'_j$ by a piecewise causal curve $\beta:[0,\frac{1}{j}]\rightarrow [0,1]\times \mathbb{D}^n$, $\beta_j(s)=(k_j(s),\alpha(s))$ with at most one break so that $k_j(s)=\pm j^{\alpha}$ and $|\alpha'|_{\sigma}=1$. We note that $\hat{L}_{t,g_j}(\beta_j)\le j^{\alpha-1}$ in order to notice
       \begin{align}
           \hat{d}_{t,g_j}(p'_j,q'_j)&\le \hat{d}_{t,g_j}(p'_j,p''_j)+\hat{d}_{t,g_j}(p''_j,q'_j) \le |t(p'_j)+t(q'_j)|+2j^{\alpha-1}
       \end{align}
      Notice that $|t(p'_j)+t(q'_j)| \le \frac{2}{j}$ since they are points on $[0,1] \times \partial B_{\frac{1}{j}}(0)$ which can be connected by a Minkowski null curve. Then, since $\hat{d}_{t,g}(p,p'_j)+\hat{d}_{t,g}(q'_j,q) \le \hat{L}_{t,g_j}(\beta)$ by construction, we see that
       \begin{align}
           \hat{d}_{t,g_j}(p,q)& \le \hat{d}_{t,g_j}(p,p'_j)+\hat{d}_{t,g_j}(p'_j,q'_j)+\hat{d}_{t,g_j}(q'_j,q) 
           \\& \le \hat{d}_{t,g_0}(p,p'_j)+|t(p'_j)+t(q'_j)|+2j^{\alpha-1}+\hat{d}_{t,g_0}(q'_j,q) 
             \\& \le \hat{L}_{t,g_j}(\beta)+\frac{2}{j}+2j^{\alpha-1}
             \\&\le \hat{d}_{t,g_0}(p,q)+\frac{2}{j}+2j^{\alpha-1},
       \end{align}
       and we are done.
  \end{proof}

\section{Proof of Main Theorem}\label{sect-Proofs}

In this section we provide the proofs of the main theorems. We start with a theorem which establishes uniform convergence form below. This should be compared with the similar estimates in \cite{Allen-Perales-Sormani} which establish uniform convergence from below. One should note that in the Riemannian case, where $g_j$ is a sequence of Riemannian metrics which we want to relate to the Remannian metric $g_{\infty}$, the result follows almost immediately since $g_j \ge (1-\frac{1}{j})g_0$ implies an estimate on the lengths of curves and hence distances. In the case of the null distance, $f_j \ge (1-\frac{1}{j})f_{\infty}$ implies a relationship between sets of admissible curves which gives the relationship for distances. In general, this is a common difference between estimating distances with respect to Riemannian manifolds versus estimating distances with respect to the null distance of Lorentzian manifolds. 

\begin{thm}\label{thm-LowerDistanceBound}
Let $(\Sigma^n,\sigma)$ be a compact, connected Riemannian manifold, $M=[t_0,t_1]\times \Sigma$, $f_j:[t_0,t_1]\rightarrow (0,\infty)$ continuous, and $g_j=-dt^2+f_j(t)^2\sigma$, $j \in \N\cup\{\infty\}$. Then if $\tau: M \rightarrow (0,\infty)$ is a differentiable time function and 
\begin{align}
f_j(t) \ge \left(1-\frac{1}{j}\right) f_{\infty}(t)
\end{align}
then
\begin{align}
    \hat{d}_{\tau,g_j} \ge  \hat{d}_{\tau,g_{\infty}}-C(j),
\end{align}
where $C(j) \ge 0$ and $C(j) \rightarrow 0$ as $j \rightarrow \infty$.
\end{thm}
\begin{proof}
Define the metric $\tilde{g}_{j,\infty}=-dt^2+\left(1-\frac{1}{j}\right)^2 f_{\infty}^2\sigma$ and notice that by assumption every piecewise causal curve with respect to $g_j$ is a piecewise causal curve with respect to $\tilde{g}_{j,\infty}$. Hence, since the null distance is defined as the infimum of the length of all piecewise causal curves and the length only depends on the time function we find that
\begin{align}
    \hat{d}_{\tau,g_j}(p,q) \ge \hat{d}_{\tau,\tilde{g}_{j,\infty}}(p,q), \quad \forall p,q \in M.
\end{align}

Now notice that by definition of a differentiable time function  $\frac{dt}{d \tau}>0$ and $\frac{d \tau}{dt}>0$. Hence we can rewrite
\begin{align}
\tilde{g}_{j,\infty}= -\left(\frac{dt}{d \tau} \right)^2 d\tau^2 +\left(1-\frac{1}{j}\right)^2 f_{\infty}^2\sigma   
\end{align}
and also define the conformal metrics
\begin{align}
\hat{g}_{j,\infty}=\left(\frac{d \tau}{dt} \right)^2\tilde{g}_{j,\infty}&= - d\tau^2 +\left(\frac{d \tau}{dt} \right)^2\left(1-\frac{1}{j}\right)^2 f_{\infty}^2\sigma,
\\\hat{g}_{\infty}=\left(\frac{d \tau}{dt} \right)^2\tilde{g}_{\infty}&= - d\tau^2 +\left(\frac{d \tau}{dt} \right)^2 f_{\infty}^2\sigma.
\end{align}
Then we notice that $\tilde{g}_{j,\infty} \rightarrow g_{\infty}$ and $\hat{g}_{j,\infty} \rightarrow \hat{g}_{\infty}$  uniformly so by Theorem \ref{thm:fuconv}, which is applicable since we can take $\tau$ as the coordinate time for $\hat{g}_{j,\infty}$ and $g_{\infty}$, we see that
\begin{align}
    \hat{d}_{\tau,\hat{g}_{j,\infty}} \rightarrow \hat{d}_{\tau,\hat{g}_{\infty}},
\end{align}
uniformly. Then by the conformal invariance of the null distance, Theorem \ref{thm-BasicNullProperties}, we know
\begin{align}
    \hat{d}_{\tau,\tilde{g}_{j,\infty}} \rightarrow \hat{d}_{\tau,g_{\infty}},
\end{align}
which implies the desired result.
\end{proof}

Now we give the proof of Theorem \ref{thm-MainTheorem tau} where we can take advantage of the conformal invariance of the null distance, Theorem \ref{thm-BasicNullProperties}, to simplify the argument.

\begin{proof}[Proof of Theorem \ref{thm-MainTheorem tau}]
 We can modify the following metrics,
\begin{align}
    g_j&=-dt^2+f_j(t)^2\sigma,
\\g_{\infty}&=-dt^2+f_{\infty}(t)^2 \sigma,
\end{align}
by multiplying by a conformal factor

\begin{align}
    \tilde{g}_j&=\frac{1}{f_{\infty}(t)^2}g_j=-\frac{1}{f_{\infty}(t)^2}dt^2+\frac{f_j(t)^2}{f_{\infty}(t)^2}\sigma,
\\\tilde{g}_{\infty}&=\frac{1}{f_{\infty}(t)^2}g_{\infty}=-\frac{1}{f_{\infty}(t)^2}dt^2+ \sigma.
\end{align}
If we define a new time function
\begin{align}
    \tau(t)&= \int_{t_0}^{t} \frac{1}{f_{\infty}(r)}dr, \quad \tau \in [0,S], t \in [t_0,t_1],
\end{align}
then we can rewrite
\begin{align}
    \tilde{g}_j&=-d\tau^2+\left(\frac{f_j(\tau)}{f_{\infty}(\tau)}\right)^2\sigma,
\\\tilde{g}_{\infty}&=-d\tau^2+ \sigma,
\end{align}
and note that by the conformal invariance of the null distance, Theorem \ref{thm-BasicNullProperties}, we know
\begin{align}
    \hat{d}_{\tau,g_j}&= \hat{d}_{\tau,\tilde{g}_j},
    \\ \hat{d}_{\tau,g_{\infty}}&= \hat{d}_{\tau,\tilde{g}_{\infty}}.
\end{align}
Now if we define $\tilde{f}_j(\tau)=\frac{f_j(\tau)}{f_{\infty}(\tau)} \ge 1$ and we know that $f_j(t)$ converges to $f_{\infty}(t)$ in $L^1$ then by the fact that $d\tau=\frac{1}{f_{\infty}(t)}dt$ and $0<c<f_{\infty}<\infty$ we know that $\tilde{f}_j(\tau)$ converges to $1$ in $L^1$ since
\begin{align}
    \int_{t_0}^{t_1}|f_j-f_{\infty}|dt&=\int_{t_0}^{t_1}\left|\frac{f_j}{f_{\infty}}-1\right|f_{\infty}dt 
    \\&= \int_0^S |\tilde{f}_j-1|f_{\infty}^2d\tau \ge c^2\int_0^S |\tilde{f}_j-1| d\tau.
\end{align}

Now, consider a  null curve with respect to $\tilde{g}_{\tau,\infty}$, $\bar{\beta}(s)=(s,\alpha(s))$ where $|\alpha'(s)|_{\sigma}=1$. We would like to approximate that curve by a piecewise null curve $\bar{\beta}_j(s)=(h_j(s),\alpha(s))$ so that  $|h_j'(s)|=\tilde{f}_j(s)\ge 1$ and $|\alpha'(s)|_{\sigma}=1$. We would like to choose $\bar{\beta}_j$ to have as few breaks as necessary and such that $\bar{\beta}_j$ remains in $M$. On a segment $[a,b] \subset [0,S]$ for which $\bar{\beta}_j$ is null we can calculate 
\begin{align}
\hat{L}_{g_j,\tau}(\bar{\beta}_j|_{[a,b]}) &= \int_a^b \tilde{f}_j(r)dr \rightarrow |b-a|.\label{LengthCurveConverges}
\end{align}
We also know that for any  $[a,b] \subset [0,S]$
\begin{align}
    \left|\int_a^b \tilde{f}_j(r)dr-\int_a^b 1dr \right| \le \int_a^b|\tilde{f}_j(r)-1|dr \le \int_0^S|\tilde{f}_j(r)-1|dr \rightarrow 0,
\end{align}
and hence the convergence in \eqref{LengthCurveConverges} is uniform for any $[a,b] \subset [0,S
]$. 

Hence $|t(\bar{\beta}(S))-t(\bar{\beta}_j(S))| \le C(j)$ where $C(j) \rightarrow 0$ as $j \rightarrow \infty$. So we can extend $\bar{\beta}_j$ by moving solely in the $t$-direction at most $C(j)$ so that $\bar{\beta}$ and $\bar{\beta}_j$ connect the same endpoints. 

Hence we see that 
\begin{align}
    \hat{d}_{\tau, g_j}(\bar{\beta}(a),\bar{\beta}(b))& \le \hat{L}_{\tau, g_j}(\bar{\beta}_j) 
    \\&= |b-a|+C(j) = \hat{d}_{\tau, g_{\infty}}(\bar{\beta}(a),\bar{\beta}(b))+C(j).
\end{align}

Now consider a piecewise null curve $\beta$ which realizes the distance between $p$ and $q$ with respect to $\tilde{g}_{\infty}$. One should choose $\beta$ with the minimum number of null curve pieces which can be bounded by a constant $N:=N(S, \diam(\Sigma,\sigma))\in \N$ which does not depend on $p,q$. If for each null curve piece $\bar{\beta}$  of $\beta$,  we replace it by $\bar{\beta}_j$ as above to obtain $\beta_j$, then we see that
\begin{align}
 \hat{d}_{\tau,\tilde{g}_j}(p,q)&\le    \hat{L}_{\tau,\tilde{g}_{j}}(\beta_j)
\le \hat{d}_{\tau,\tilde{g}_{\infty}}(p,q)+NC(j).
\end{align}

By the assumption that $\tilde{f}_j \ge 1$ and Theorem \ref{thm-LowerDistanceBound} we have that 
\begin{align}
\hat{d}_{\tau,\tilde{g}_{\infty}}(p,q) \le  \hat{d}_{\tau,\tilde{g}_j}(p,q),
\end{align}
and hence we are done.
\end{proof}

Now we give the proof of Theorem \ref{thm-MainTheorem t} which uses Theorem \ref{thm-MainTheorem tau} to establish uniform convergence of $\hat{d}_{\tau,g_j}$ and then the main work of the argument is to show that this also implies uniform convergence of $\hat{d}_{t,g_j}$.

\begin{proof}[Proof of Theorem \ref{thm-MainTheorem t}]

Consider
\begin{align}
    \tilde{g}_j = -dt^2+\left(1-\frac{1}{j}\right)^{-2}f_j(t)^2\sigma =-dt^2+\tilde{f}_j(t)^2\sigma
\end{align}
so that
\begin{align}
\tilde{f}_j(t)= \left(1-\frac{1}{j}\right)^{-1}f_j(t) &\ge f_{\infty}(t).
\end{align}
By Theroem \ref{thm-MainTheorem tau}  
\begin{align}
    \tau(t)=\int_{t_0}^t\frac{1}{f_{\infty}(r)}dr, \quad t \in [t_0,t_1],
\end{align}
and hence
\begin{align}
    \hat{d}_{\tau,\tilde{g}_j} \rightarrow \hat{d}_{\tau,g_{\infty}},
\end{align}
uniformly. Now by Theorem \ref{thm-LowerDistanceBound} and the fact that
\begin{align}
 \left(1-\frac{1}{j}\right)f_{\infty}(t) \le f_j(t)\le \left(1-\frac{1}{j}\right)^{-1}f_j(t)=\tilde{f}_j(t),  
\end{align} 
we find
\begin{align}
 \hat{d}_{\tau,g_{\infty}}-C(j) \le  \hat{d}_{\tau,g_j} \le \hat{d}_{\tau,\tilde{g}_j},
\end{align}
and hence
\begin{align}\label{UniformConvergenceProof2}
     \hat{d}_{\tau,g_j} \rightarrow \hat{d}_{\tau,g_{\infty}},
\end{align}
uniformly as well.

   Now by Lemma \ref{lem:timeequiv} we know that
    \begin{align}\label{UniformBoundProof2}
     \frac{1}{C} \hat{d}_{\tau,g_j}(p,q)   \le \hat{d}_{t,g_j}(p,q) \le C \hat{d}_{\tau,g_j}(p,q).
    \end{align}

Now we would like to show that the Arzel\`{a}-Ascoli Theorem  applies to the sequence $\hat{d}_{t,g_j}$. To this end, fix $\varepsilon>0$ and note that by \eqref{UniformConvergenceProof2} there exists a $J \in \N$ so that for $j > J$ we find
\begin{align}
    |\hat{d}_{\tau,g_j}(p,q)-\hat{d}_{\tau,g_{\infty}}(p,q)|<\frac{\varepsilon}{2C}.
\end{align}

Since $\hat{d}_{\tau,g_j}-\hat{d}_{\tau,g_{\infty}}$ is continuous, for each $j \le J$ we can choose a $\delta_j$ so that for $p,q \in B_{\hat{d}_{\tau,g_{\infty}}}(x,\delta_j)=\{p\in M:\hat{d}_{\tau,g_{\infty}}(x,p)\le \delta_j\}$ we find
\begin{align}
    |\hat{d}_{\tau,g_j}(p,q)-\hat{d}_{\tau,g_{\infty}}(p,q)|<\frac{\varepsilon}{4C}.
\end{align}
If we set $\delta=\min\{\delta_1,...,\delta_J,\frac{\varepsilon}{4C}\}$ we find for $p \in B_{\hat{d}_{\tau,g_{\infty}}}(x,\delta)=\{p\in M:\hat{d}_{\tau,g_{\infty}}(x,p)\le \delta\}$ that
\begin{align}
 \hat{d}_{t,g_j}(x,p)&\le   C \hat{d}_{\tau,g_j}(x,p)
 \\&\le C\hat{d}_{\tau,g_{\infty}}(x,p)+C|\hat{d}_{\tau,g_j}(x,p)-\hat{d}_{\tau,g_{\infty}}(x,p)|
 \\&\le C(\delta)+C\left(\frac{\varepsilon}{4C}\right) \le \frac{\varepsilon}{2}.
\end{align}
Now for $x,y \in M$ and  $(x',y') \in B_{\hat{d}_{\tau,g_{\infty}}}(x,\delta)\times B_{\hat{d}_{\tau,g_{\infty}}}(y,\delta)$, by the triangle inequality we find
\begin{align}
    |\hat{d}_{t,g_j}(x,y)-\hat{d}_{t,g_j}(x',y')| \le \hat{d}_{t,g_j}(x,x')+\hat{d}_{t,g_j}(y,y') \le \frac{\varepsilon}{2}+\frac{\varepsilon}{2}=\varepsilon.
\end{align}
Hence we see that $\hat{d}_{t,g_j}$ is equibounded and equicontinuous and hence by the Arzel\`{a}-Ascoli Theorem we know that a subsequence converges uniformly to a function $d_0:M\times M\rightarrow [0,\infty)$ which is symmetric and satisfies the triangle inequality.

 We know from Theorem \ref{thm-LowerDistanceBound} that
    \begin{align}\label{LowerDistanceFact}
        \hat{d}_{t,g_{\infty}}-C(j) \le \hat{d}_{t,g_j},
    \end{align}
    and hence we see that $\hat{d}_{t,g_{\infty}} \le d_0$ and $d_0$ is positive definite. 
    
   Now we will estimate $\hat{d}_{t,g_j}$ in terms of $\hat{d}_{t,g_{\infty}}$ and show pointwise convergence to $\hat{d}_{t,g_{\infty}}$. To this end, consider $p,q \in M$ so that $q \in J^{\pm}_{g_{\infty}}(p)$ and hence
\begin{align}
    \hat{d}_{\tau,g_{\infty}}(p,q)&=|\tau(p)-\tau(q)|,
    \\ \hat{d}_{t,g_{\infty}}(p,q)&=|t(p)-t(q)|.
\end{align}
Now by \eqref{UniformConvergenceProof2} we know that
\begin{align}\label{relateEstimate}
    \hat{d}_{\tau,g_j}(p,q)& \le|\tau(p)-\tau(q)|+C'(j),
\end{align}
where $C'(j) \rightarrow 0$ as $j \rightarrow \infty$ and we want to relate \eqref{relateEstimate} to $\hat{d}_{t,g_j}$.

Let $\beta_j:[0,1]\rightarrow M$ be a piecewise null curve with respect to $g_j$ so that
\begin{align}
    \hat{L}_{\tau,g_j}(\beta_j) \le \hat{d}_{\tau,g_j}(p,q)+\frac{1}{j} \le |\tau(p)-\tau(q)|+\frac{1}{j}+C'(j).\label{AlmostMinimizingEq}
\end{align}
Choose $I_j \subset [0,1]$ to be a union of finitely many intervals $N_j \in \N$ so that for all $t \in I_j$ we find $\beta_j'(t)$ is always future pointing or always past pointing and 
\begin{align}
\hat{L}_{\tau,g_j}(\beta_j|_{I_j})=|\tau(p)-\tau(q)|. \label{WellChosenSegmentEq} 
\end{align}
This can be done since $\tau$ is increasing along all future pointing curves and $\beta_j$ is a continuous curve connecting $p$ to $q$.
Note, since $t$ is also increasing along all future pointing curves, \eqref{WellChosenSegmentEq} also implies that
\begin{align}
\hat{L}_{t,g_j}(\beta_j|_{I_j})=|t(p)-t(q)|.  \label{WellChosenSegmentEq2} 
\end{align}
When we combine \eqref{AlmostMinimizingEq} with \eqref{WellChosenSegmentEq} we find
\begin{align}
   \hat{L}_{\tau,g_j}(\beta_j)&= \hat{L}_{\tau,g_j}(\beta_j|_{I_j})+\hat{L}_{\tau,g_j}(\beta_j|_{I_j^c}) 
   \\&= |\tau(p)-\tau(q)|+\hat{L}_{\tau,g_j}(\beta_j|_{I_j^c}) \le |\tau(p)-\tau(q)|+\frac{1}{j}+C'(j),
\end{align}
and hence
\begin{align}
   \hat{L}_{\tau,g_j}(\beta_j|_{I_j^c}) \le  \frac{1}{j}+C'(j).\label{PastPiecesEstimate}
\end{align}
Decompose $\beta_j|_{I_j^c}$ into the minimum number of pieces $\beta_j^1,...,\beta_j^P$, $P \in \N$, with break points $p_j^0,p_j^1,...,p_j^P$, so that each curve is either a continuous future or past null curve. Since $\tau$ is anti-Lipschitz (See Definition \ref{def-AntiLipschitz} and Lemma 4.10 and Proposition 4.12 of \cite{SV}) we know that for $g_0=dt^2+\sigma$, $K>0$ we can estimate
\begin{align}
    \hat{L}_{\tau,g_j}(\beta_j|_{I_j^c})= \sum_{i=1}^P |\tau(p_j^{i+1})-\tau(p_j^{i})| \ge K \sum_{i=1}^P d_{g_0}(p_j^{i+1},p_j^{i}).\label{AntiLipschitzPastPieceEstimate}
\end{align}
When \eqref{PastPiecesEstimate} is combined with \eqref{AntiLipschitzPastPieceEstimate} we find
\begin{align}
    \hat{L}_{t,g_j}(\beta_j|_{I_j^c}) &= \sum_{i=1}^P |t(p_j^{i+1})-t(p_j^{i})| 
    \\&\le  \sum_{i\in I_p} d_{g_0}(p_j^{i+1},p_j^{i}) \le \frac{1}{K}\left(\frac{1}{j}+C'(j) \right), \label{EvenCurveEstimate}
\end{align}
where we have used that $t$ is Lipschitz with respect to $g_0$. Now we can apply \eqref{WellChosenSegmentEq2} and \eqref{EvenCurveEstimate} to estimate
\begin{align}
   \hat{d}_{t,g_j}(p,q)&\le\hat{L}_{t,g_j}(\beta_j)  
  \\& = \hat{L}_{t,g_j}(\beta_j|_{I_j})+\hat{L}_{t,g_j}(\beta_j|_{I_j^c})
   \\&\le|t(p)-t(q)|+\frac{1}{K}\left(\frac{1}{j}+C'(j) \right).\label{AlmostDifferencint-values}
\end{align}

Now consider any $p,q \in M$ and for $\varepsilon>0$ let $\beta$ be a piecewise null curve with respect to $g_{\infty}$ so that $\hat{d}_{t,g_{\infty}}(p,q)\ge\hat{L}_{t,g_{\infty}}(\beta)-\varepsilon$. Let $p=p_0,p_1,p_2,...,p_{k-1},p_k=q$ be the minimal number of points so that $\beta$ connects these points via null curves. Then we see by the triangle inequality that
\begin{align}
 \hat{d}_{t,g_j}(p,q)& \le \sum_{i=1}^k   \hat{d}_{t,g_j}(p_{i-1},p_i), \label{DistanceConsequence} 
\end{align}
and by taking the limsup on both sides of \eqref{DistanceConsequence} and applying \eqref{AlmostDifferencint-values} we find
\begin{align}
\limsup_{j\rightarrow \infty} \hat{d}_{t,g_j}(p,q)& \le \limsup_{j\rightarrow \infty}\sum_{i=1}^k   \hat{d}_{t,g_j}(p_{i-1},p_i)
\\& \le \sum_{i=1}^k\limsup_{j\rightarrow \infty}   \hat{d}_{t,g_j}(p_{i-1},p_i)
\\& =\sum_{i=1}^k   |t(p_{i-1})-t(p_i)|
\\& = \hat{L}_{t,g_{\infty}}(\beta) \le \hat{d}_{t,g_{\infty}}(p,q)+ \varepsilon.
\end{align}
Since this is true for all $\varepsilon>0$ we have the desired result and when combined with \eqref{LowerDistanceFact} we see that $\hat{d}_{t,g_j}$ converges pointwise to $\hat{d}_{t,g_{\infty}}$. By the uniform convergence established earlier, we see that $d_0=\hat{d}_{t,g_{\infty}}$. Furthermore, the subsequence taken when applying the Arzel\`{a}-Ascoli theorem was not needed since any subsequence which supposedly did not converge to $\hat{d}_{t,g_{\infty}}$ would satisfy the same hypotheses as the original sequence and hence would also have a subsequence converging to $\hat{d}_{t,g_{\infty}}$. Hence $\hat{d}_{t,g_j}$ converges uniformly to $\hat{d}_{t,g_{\infty}}$, as desired.
\end{proof}

\bibliography{bibliography}

\begin{bibdiv}
\begin{biblist}

\bib{AB}{article}{
      author={Allen, Brian},
      author={Burtscher, Annegret},
       title={{Properties of the Null Distance and Spacetime Convergence}},
        date={2021},
        ISSN={1073-7928},
     journal={International Mathematics Research Notices},
      volume={2022},
      number={10},
       pages={7729\ndash 7808},
}

\bib{AGH}{article}{
      author={Andersson, L.},
      author={Galloway, G.J.},
      author={Howard, R.},
       title={The cosmological time function},
        date={1998},
     journal={Classical Quantum Gravity},
      volume={15},
      number={2},
       pages={309\ndash 322},
}

\bib{Allen-Perales-Sormani}{article}{
      author={Allen, Brian},
      author={Perales, Raquel},
      author={Sormani, Christina},
       title={Volume above distance below},
        date={2020},
     journal={arXiv:2003.01172 [math.MG]},
}

\bib{Allen-Sormani}{article}{
      author={Allen, Brian},
      author={Sormani, Christina},
       title={Contrasting various notions of convergence in geometric
  analysis},
        date={2019},
     journal={Pacific Journal of Mathematics},
      volume={303},
      number={1},
       pages={1\ndash 46},
}

\bib{Allen-Sormani-2}{article}{
      author={Allen, Brian},
      author={Sormani, Christina},
       title={Relating notions of convergence in geometric analysis},
        date={2020},
     journal={Nonlinear Analysis},
      volume={200},
}

\bib{BBI}{book}{
      author={Burago, D.},
      author={Burago, Y.},
      author={Ivanov, S.},
       title={A course in metric geometry},
      series={Graduate Studies in Mathematics},
   publisher={American Mathematical Society, Providence, RI},
        date={2001},
      volume={33},
        ISBN={0-8218-2129-6},
      review={\MR{1835418}},
}

\bib{BG}{misc}{
      author={Burtscher, Annegret},
      author={García-Heveling, Leonardo},
       title={Time functions on lorentzian length spaces},
        date={2021},
}

\bib{BG2}{misc}{
      author={Burtscher, Annegret},
      author={García-Heveling, Leonardo},
       title={Global hyperbolicity through the eyes of the null distance},
        date={2023},
}

\bib{B}{misc}{
      author={Braun, M.},
       title={R\'enyi's entropy on lorentzian spaces. timelike
  curvature-dimension conditions},
        date={2022},
}

\bib{Bernal-Sanchez}{article}{
      author={Bernal, Antonio~N.},
      author={S\'{a}nchez, Miguel},
       title={Smoothness of time functions and the metric splitting of globally
  hyperbolic spacetimes},
        date={2005},
        ISSN={0010-3616,1432-0916},
     journal={Comm. Math. Phys.},
      volume={257},
      number={1},
       pages={43\ndash 50},
}

\bib{CM}{article}{
      author={Cavalletti, F.},
      author={Mondino, A.},
       title={A review of lorentzian synthetic theory of timelike ricci
  curvature bounds},
        date={2022},
     journal={General Relativity and Gravitation},
      volume={54},
}

\bib{GS}{inproceedings}{
      author={Graf, Melanie},
      author={Sormani, Christina},
       title={Lorentzian area and volume estimates for integral mean curvature
  bounds},
        date={2022},
   booktitle={Developments in lorentzian geometry},
      editor={Albujer, Alma~L.},
      editor={Caballero, Magdalena},
      editor={Garc{\'i}a-Parrado, Alfonso},
      editor={Herrera, J{\'o}natan},
      editor={Rubio, Rafael},
   publisher={Springer International Publishing},
}

\bib{SH}{article}{
      author={Harris, Steven~G.},
       title={A triangle comparison theorem for {L}orentz manifolds},
        date={1982},
        ISSN={0022-2518,1943-5258},
     journal={Indiana Univ. Math. J.},
      volume={31},
      number={3},
       pages={289\ndash 308},
}

\bib{KCS}{article}{
      author={Kunzinger, M.},
      author={Sämann, C.},
       title={Lorentzian length spaces},
        date={2018},
     journal={Annals of Global Analysis and Geometry},
      volume={54},
}

\bib{KS}{article}{
      author={Kunzinger, Michael},
      author={Steinbauer, Roland},
       title={Null distance and convergence of lorentzian length spaces},
        date={2022},
     journal={Annales Henri Poincaré},
      volume={23},
}

\bib{M}{misc}{
      author={Müller, Olaf},
       title={Lorentzian gromov-hausdorff theory and finiteness results},
        date={2022},
}

\bib{MS}{misc}{
      author={Minguzzi, E.},
      author={Suhr, S.},
       title={Lorentzian metric spaces and their gromov-hausdorff convergence},
        date={2023},
}

\bib{N}{article}{
      author={Noldus, J.},
       title={A lorentzian gromov-hausdorff notion of distance},
        date={2004},
        ISSN={0264-9381},
     journal={Classical Quantum Gravity},
      volume={21},
      number={4},
       pages={839\ndash 850},
}

\bib{S-Ob}{article}{
      author={Sormani, C.},
       title={Spacetime intrinsic flat convergence},
        date={2018},
     journal={Oberwolfach Report for the Workshop ID 1832: Mathematical General
  Relativity},
       pages={1\ndash 3},
}

\bib{SSFuture}{article}{
      author={Sakovich, A.},
      author={Sormani, C.},
       title={Future work},
}

\bib{SS}{article}{
      author={Sakovich, A.},
      author={Sormani, C.},
       title={The null distance encodes causality},
        date={2023},
     journal={Journal of Mathematical Physics},
      volume={64},
      number={1},
}

\bib{SV}{article}{
      author={Sormani, C.},
      author={Vega, C.},
       title={Null distance on a spacetime},
        date={2016},
        ISSN={0264-9381},
     journal={Classical Quantum Gravity},
      volume={33},
      number={8},
       pages={085001, 29},
      review={\MR{3476515}},
}

\bib{V}{misc}{
      author={Vega, Carlos},
       title={Spacetime distances: an exploration},
        date={2021},
}

\end{biblist}
\end{bibdiv}

\end{document}